\documentclass[twocolumn,superscriptaddress,prd,amsmath,amsfonts,nofootinbib,showpacs]{revtex4-1}

\newcommand{\ud}{\text{d}}
\newcommand{\half}{{\tfrac{1}{2}}}
\newcommand{\gb}{\check{g}}
\newcommand{\bs}{\begin{split}}
\newcommand{\es}{\end{split}}
\newcommand{\delb}{\check{\nabla}}

\newcommand{\mc}[1]{\mathcal{#1}}
\newcommand{\bc}[1]{\check{\mathcal{#1}}}
\newcommand{\TT}{\textsc{tt}}
\newcommand{\hb}{\bar{h}}

\usepackage{color}
\definecolor{bwblue}{rgb}{0,0,0.7}
\usepackage{hyperref}
\hypersetup{
    colorlinks,
    citecolor=bwblue,
    filecolor=bwblue,
    linkcolor=bwblue,
    urlcolor=bwblue,
    pdftitle={Localising the Energy and Momentum of Linear Gravity},
    pdfauthor={Luke M. Butcher}}
\usepackage{slashed}
\usepackage{graphicx}
\usepackage{subfigure}
\usepackage{psfrag}
\usepackage{amsthm}
\theoremstyle{plain}
\newtheorem*{Theorem}{Theorem}

\begin{document}

\title{Localising the Energy and Momentum of Linear Gravity}
\author{Luke M. Butcher}
\email[]{l.butcher@mrao.cam.ac.uk}
\affiliation{Astrophysics Group, Cavendish Laboratory, J J Thomson Avenue, Cambridge, CB3 0HE, UK}
\affiliation{Kavli Institute for Cosmology, Madingley Road, Cambridge, CB3 0HA, UK}
\author{Michael Hobson}
\affiliation{Astrophysics Group, Cavendish Laboratory, J J Thomson Avenue, Cambridge, CB3 0HE, UK}
\author{Anthony Lasenby}
\affiliation{Astrophysics Group, Cavendish Laboratory, J J Thomson Avenue, Cambridge, CB3 0HE, UK}
\affiliation{Kavli Institute for Cosmology, Madingley Road, Cambridge, CB3 0HA, UK}
\date{18 November 2010}
\pacs{04.20.Cv, 04.30.-w}

\begin{abstract}
A framework is developed which quantifies the local exchange of energy and momentum between matter and the linearised gravitational field. We derive the unique gravitational energy-momentum tensor consistent with this description, and find that this tensor only exists in the harmonic gauge. Consequently, nearly all the gauge freedom of our framework is naturally and unavoidably removed. The gravitational energy-momentum tensor is then shown to have two exceptional properties: (a) it is gauge-invariant for gravitational plane-waves, (b) for arbitrary transverse-traceless fields, the energy-density is never negative, and the energy-flux is never spacelike. We analyse in detail the local gauge invariant energy-momentum transferred between the gravitational field and an infinitesimal point-source, and show that these invariants depend only on the transverse-traceless components of the field. As a result, we are led to a natural gauge-fixing program which at last renders the energy-momentum of the linear gravitational field completely unambiguous, and additionally ensures that gravitational energy is never negative nor flows faster than light. Finally, we calculate the energy-momentum content of gravitational plane-waves, the linearised Schwarzschild spacetime (extending to arbitrary static linear spacetimes) and the gravitational radiation outside two compact sources: a vibrating rod, and an equal-mass binary.
\end{abstract}
\maketitle

\section{Introduction}
Half a century ago, a simple argument established that gravitational waves carry energy and can exchange this energy with matter. Often attributed to Feynman (certainly popularised by Bondi \cite{Bondi}) the argument asked us to imagine a gravitational detector comprising a rigid rod along which two ``sticky beads'' are threaded. A passing gravitational wave then acts to alter the proper distance between the beads, and this motion, opposed by friction, heats the detector and thus mediates a transfer of energy from gravity to matter. Despite the simplicity of this idea, even after fifty years, it has not been possible to explain \emph{where} in spacetime this gravitational energy resides, and it is generally accepted that attempts to do so are ``looking for the right answer to the wrong question''\cite{MTW}.

The elusiveness of the ``right answer'', and the wrongness of the question, are very often identified as arising from gravity's gauge freedom, the consequence of which is a one-to-many mapping between physical spacetime and whatever localisation of gravitational energy-momentum might be proposed. Historically this issue was cast in terms of coordinate dependence, and the multitude of non-covariant objects that were constructed (first by Einstein \cite{Einstein}, and most famously by Landau and Lifshitz \cite{LL}) were termed energy-momentum \emph{pseudotensors}. However, a more recent formulation \cite{BG} has made it clear that the construction of a genuine tensor (defined on some background spacetime) is not the central problem; rather, it is the tensor's dependence on the arbitrary diffeomorphism that maps physical spacetime to the background \cite{Butcher08}. 

Nevertheless, there is no reason \emph{a priori} that gauge dependence should preclude the construction of a physically unambiguous tensor, provided we are prepared to remove the gauge freedom in some well-defined way. In cosmology this is frequently done by constructing new variables which are gauge invariant but equal to the relevant gauge-dependent fields (such as gravity or density fluctuations) in a particular gauge \cite{Bardeen,Stewart}; however, it is just as effective to provide a physically unambiguous method by which the gauge may be fixed, and to then insist that the gravitational field be evaluated in this gauge when locating its energy and momentum. Unfortunately, no previous approach has supplied instructions of this nature, and more importantly, neither the construction of these energy-momentum objects, nor their key properties, appear to favour one gauge (or one set of gauge-invariants) over another; thus it appears impossible to justify any of these seemingly arbitrary choices as \emph{natural}.

Besides gauge dependence, there is also a great range of choice over which properties, physical or mathematical, should define the gravitational energy-momentum tensor: should we be guided by a putative conservation law, or have in mind a particular role in the field equations? For instance, it is always possible to locate the energy-momentum of \emph{matter} by measuring the gravity it generates, so one might suggest that gravity's energy-momentum should be localised in a similar fashion, by examining the interaction it has with itself. Following this idea to its conclusion, it  has been shown \cite{Deser, feyngrav,Butcher09,Redux} that general relativity may be constructed from an initially linear (spin-2) field theory that is then systematically coupled to its own (Hilbert) energy-momentum tensor. Sadly, this scheme leads us to identify the non-linear part of the Einstein tensor $G_{ab}-G^{(1)}_{ab}$ as the gravitational energy-momentum, so (a) the gauge problem remains, and (b) the result is additionally ambiguous, as different choices of ``gravitational field'' ($g_{ab}$, $g^{ab}$, $\sqrt{-g}g^{ab}$, etc.)  mix the linear and non-linear terms in $G_{ab}$.

In spite of these various difficulties, one aspect of this enduring problem stands opposed to conventional wisdom and motivates our present discussion: when gravity and matter interact, the \emph{exchange} of energy is local! To see this we need look no further than the sticky bead detector: here, the energy exchange is certainly localised in so far as it takes place only within the confines of the detector. Furthermore, we can imagine a very small detector, much smaller than a wavelength of the incident gravitational radiation, and observe that at each instant a well-defined power is developed in the detector as heat; thus, at least in this case, the rate of energy exchange is associated with a particular point in spacetime. One might hope, therefore, that consistency with this phenomenon would be enough to localise the energy and momentum of the gravitational field outside the detector, or even when no detector is present. Moreover, even if a gravitational energy-momentum tensor could not be found, there would still be great value in constructing a framework for the description and analysis of local gravitational energy-momentum exchange. The purpose of this article is to develop precisely this framework, and to examine the gravitational energy-momentum tensor it brings to light. In doing so we uncover a simple and unambiguous ``right answer'' through which the effects of gravitational energy-momentum may be usefully understood. Conceivably, this was the ``right question'' to ask.

For the sake of simplicity, we have restricted our present discussion to \emph{linearised} general relativity on a flat Minkowski background. It is only in this linear regime that the convenient fiction of a ``gravitational field'' propagating on a background spacetime can be taken seriously, a construction which is essentially unavoidable when localising gravitational energy-momentum.\footnote{As long as there is some spacetime with everywhere vanishing gravitational energy-momentum, then this will naturally play the role of the background, and fluctuations away from this configuration will constitute the gravitational field. Although the most natural choice for this ``ground-state'' is flat spacetime,  this does not necessarily preclude the extension of our formalism to less trivial backgrounds; however, we suspect there may be technical or conceptual problems with ``ignoring'' the energy-momentum of a nontrivial background. In particular, we anticipate issues analogous to those of associating energy-momentum with a fluctuation in the electromagnetic field $\delta F_{ab}$ when the background $\check{F}_{ab}$ is non-zero: the energy-momentum tensor $T\sim \check{F}^2 + \check{F}\delta F+ (\delta F)^2$, so the dominant contribution from the fluctuation will be linear in the field, rather than quadratic.}  On a technical level, the restriction to the linear approximation limits the space of gauge transformations to a manageable size, facilitating the analysis and eventual removal of our description's gauge dependence. Furthermore, our gravitational energy-momentum tensor will not be derived from non-linear terms in the field equations, so we avoid any ambiguity arising from field redefinition. We shall not attempt to extend our results beyond the linear theory at this time.\footnote{Of course, it may not be possible to extend the framework we develop here to the full non-linear theory, and we accept that localising gravitational energy-momentum in this regime (where the distinction between background and fluctuation is virtually meaningless) may be an inherently flawed idea. Of course, this does not alter the validity of our work in the linear case, where the ``field theoretic'' view is justified.}

The structure of the paper will be as follows. We begin by building the foundations of our framework, deriving a gravitational energy-momentum tensor (\ref{taubar}) by demanding consistency with the energy and momentum exchanged with matter. As we will see, most of the tensor's gauge freedom is eliminated immediately as a natural consequence of this derivation. Following this, we demonstrate two important additional properties of our tensor, further solidifying its interpretation as gravity's energy-momentum tensor. We then develop our framework more concretely by analysing the transfer of gravitational energy-momentum onto an infinitesimal detector; in the process of making this analysis gauge invariant, we will purge the last trace of gauge ambiguity from our energy-momentum tensor. Finally, we examine the gravitational energy-momentum in some specific examples. Throughout, we work in units where $c=1$, write $\kappa \equiv 8\pi G$, and use the sign conventions of Wald \cite{Wald}: the metric signature is $(-,+,+,+)$, and the Riemann and Ricci tensors are defined by $[\nabla_c,\nabla_d]v^a \equiv R^{a}_{\phantom{a}bcd}v^b$, and $R_{ab}\equiv R^{c}_{\phantom{c}acb}$.

\section{Motivation and Derivation}\label{Deriv}
The purpose of this section is to explain how, by considering the energy-momentum transferred between matter and gravity, we are led to a formula for the gravitational energy-momentum tensor. We begin by laying down some mathematical groundwork.

\subsection{Preliminaries}
As previously explained, this paper focuses exclusively on \emph{linear} gravity: we only consider physical spacetimes $(\mc{M},g_{ab})$ in which the curvature $R^a_{\phantom{a}bcd}$ is everywhere small. As usual, this allows us to identify the physical spacetime with a \emph{flat} background spacetime $(\check{\mc{M}},\check{g}_{ab})$, where $\check{R}^a_{\phantom{a}bcd}=0$, using a diffeomorphism \raisebox{-1pt}{$\phi: \mc{M} \rightarrow \check{\mc{M}}$}. The ``gravitational field'' $h_{ab}$ is then defined on $\check{\mc{M}}$ by
\begin{align}\label{hdef}
\phi^* g_{ab}= \gb_{ab} + h_{ab},
\end{align}
and we insist that $\phi$ be chosen such that $h_{ab}$ is small everywhere, in order that the \emph{linearised} Einstein field equations are a good approximation:\footnote{We use $O(h^n)$ as an abbreviation of $O((h_{ab})^n)$; this should not be confused with the trace of the gravitational field $h\equiv h_{ab}\gb^{ab}$.}
\begin{align}\label{FEqs}
\widehat{G}_{ab}^{\phantom{ab}cd}h_{cd}= \kappa \check{T}_{ab} + O(h^2).
\end{align}
In the above relation, $\check{T}_{ab}\equiv \phi^* T_{ab}= O(h)$ is the matter energy-momentum tensor $T_{ab}$ mapped onto the background, and 
\begin{align}\nonumber
\widehat{G}_{ab}^{\phantom{ab}cd}h_{cd} &\equiv \delb_c \delb_{(a}h_{b)}^{\phantom{a)}c} - \half \delb^2 h_{ab} - \half\delb_a\delb_b h
\\\label{Gdef}&\quad+ \half\gb_{ab}\left(\delb^2 h - \delb_c\delb_d h^{cd} \right)
\end{align}
is the linearised Einstein tensor $G^{(1)}_{ab}$. Our freedom of choice over $\phi$ will of course give rise to the usual gauge transformation $\delta h_{ab}=\delb_{(a} \xi_{b)}$.

On the background it will be useful to define four vectors\footnote{We use Roman letters as abstract tensor indices \citep[p. 437]{Wald} and Greek letters as numerical indices running from 0 to 3. Tensor indices of fields defined on the background are of course raised and lowered with $\gb_{ab}$.} $\{\check{e}_\mu^{\phantom{\mu}a}\}$ obeying
\begin{align}\label{edef1}
\delb_a \check{e}_\mu^b &=0,\\ \label{edef2}
\check{e}_\mu^{\phantom{\mu}a} \check{e}_{\nu a} &=\eta_{\mu\nu}, 
\end{align}
which form the basis of a Lorentz coordinate system $\{x^\mu\}$ on $\bc{M}$: $\check{e}_\mu^{\phantom{\mu}a}\equiv (\partial_\mu)^a$. From this starting point, we shall define a corresponding set of vector fields $\{e_\mu^{\phantom{\mu}a}\}$ in the physical spacetime,
\begin{align}
e_\mu^{\phantom{\mu}a} \equiv (\phi^{-1})^*\check{e}_\mu^{\phantom{\mu}a},
\end{align}
the behaviour of which will only be determined once we have fixed the gauge $\phi$, an issue to which we will return later.

\subsection{Energy-Momentum Currents}
Superficially, general relativity is a theory in which the energy and momentum of matter is always conserved:
\begin{align}\label{consT}
\nabla^a T_{ab}=0.
\end{align}
However, the sticky bead argument has already demonstrated that this is not the case; in reality, matter may gain (or lose) energy through interaction with the gravitational field. The reason for this apparent contradiction is as follows. In order to determine the energy of each part of the detector, one must first specify a timelike vector field $e_0^{\phantom{0}a}$ (the ``time direction'' conjugate to the energy) with which to form an energy current-density  $J^{a} \equiv T^a_{\phantom{a}b}e_0^{\phantom{0}b}$. The incoming gravitational wave will then prevent $e_0^{\phantom{0}b}$ from satisfying $\nabla_a e_0^{\phantom{0}b}=0$, and we will find that $\nabla_a J^{a} = T^a_{\phantom{a}b}\nabla_a e_0^{\phantom{0}b}\ne 0$. This inequality indicates a mismatch between the energy of the matter flowing into a given point, and the change in energy of the matter at that point; in other words, it represents the appearance of \emph{additional energy} which was not already present in the matter --  this is the energy absorbed from the gravitational wave! What is needed, therefore, is a framework which can account for this gained energy by identifying a corresponding loss in the energy of the gravitational field. We devote the rest of this section to the development of this idea, which will form the basis of our description of gravitational energy-momentum.

Following the previous discussion, it should now be clear that we must define one energy current-density, and three momentum current-densities, by
\begin{align}
J_\mu^{\phantom{\mu}a} \equiv T^a_{\phantom{a}b}e_\mu^{\phantom{\mu}b},
\end{align}
using the vectors $\{e_\mu^{\phantom{\mu}a}\}$ that get mapped to the Lorentz basis of the background. This is a generalisation of the practice of defining conserved currents by contracting $T_{ab}$ with a killing vector in a spacetime with a continuous symmetry. Here, however, the vector fields $\{e_\mu^{\phantom{\mu}a}\}$ only correspond to \emph{approximate} symmetries (present because spacetime is nearly flat) and thus the currents will not be conserved. The real difficulty is choosing sensible behaviour for $\{e_\mu^{\phantom{\mu}a}\}$ that sufficiently captures the ``parallelism'' of killing vectors in the absence of any gravitational symmetry. Because $e_\mu^{\phantom{\mu}a} \equiv (\phi^{-1})^*\check{e}_\mu^{\phantom{\mu}a}$, this question has been recast as a choice of gauge, which we will address later.

Having defined our energy-momentum currents (apart from specifying $\phi$) we are now in a position to express the key idea of our approach. We seek a symmetric tensor field $\tau_{ab}$, defined on the background, that is a quadratic function of the gravitational field $h_{ab}$. We wish to be able to interpret $\tau_{ab}$ as the energy-momentum tensor of the gravitational field, and we shall achieve this by insisting that its non-conservation (in the background) exactly balances the non-conservation of the $J_\mu^{\phantom{\mu}a}$ in the physical spacetime. Specifically, we wish to be able to define gravitational energy-momentum current-densities $j_\mu^{\phantom{\mu}a}$ by
\begin{align}
j_\mu^{\phantom{\mu}a} \equiv \tau^a_{\phantom{a}b}\check{e}_\mu^{\phantom{\mu}b},
\end{align}
such that
\begin{align}\label{conservation}
\delb_a j_\mu^{\phantom{\mu}a} + \phi^*(\nabla_a J_\mu^{\phantom{\mu}a}) =0.
\end{align}
This equation captures the idea that energy-momentum is transferred between matter and the gravitational field. In particular, equation (\ref{conservation}) indicates that knowing the behaviour of $h_{ab}$ at some point will be sufficient to determine the fields $\nabla_a J_\mu^{\phantom{\mu}a}$ that express the local change in energy-momentum of the matter at the corresponding point in the physical spacetime.

We proceed by calculating the two elements of (\ref{conservation}). Because we are using a Lorentz basis in the background, $\delb_a \check{e}_\mu^b =0$ trivially gives
\begin{align}\label{delj}
\delb_a j_\mu^{\phantom{\mu}a} = \check{e}_\mu^{\phantom{\mu}b}\delb_a\tau^a_{\phantom{a}b}.
\end{align}
The second term is a little less trivial; using (\ref{consT}),  
\begin{align}\nonumber
\phi^*(\nabla_a J_\mu^{\phantom{\mu}a})&= \phi^*(T^a_{\phantom{a}b}\nabla_a e_\mu^{\phantom{\mu}b}),\\\nonumber
&= \phi^*T^a_{\phantom{a}b} \phi^* (\nabla_a e_\mu^{\phantom{\mu}b}),\\
&= (\check{T}^a_{\phantom{a}b} + O(h^2))(\delb_a\check{e}_\mu^{\phantom{\mu}b} + \check{e}_\mu^{\phantom{\mu}c}C^b_{\phantom{b}ac}),
\end{align}
where $C^a_{\phantom{a}bc}= \half(\delb_b h_c^{\phantom{c}a}+ \delb_c h_b^{\phantom{b}a}- \delb^a h_{bc}) + O(h^2) $ is the connection between the two derivative operators: $\phi^*(\nabla_a v^b)= \delb_a \phi^* v^b + C^b_{\phantom{b}ac}\phi^* v^c$. Now, because  $\delb_a \check{e}_\mu^b =0$, and $\check{T}_{ab}=\check{T}_{ba}$, we have
\begin{align}\nonumber
\phi^*(\nabla_a J_\mu^{\phantom{\mu}a})&= \half \check{T}^a_{\phantom{a}b}\check{e}_\mu^{\phantom{\mu}c}(\delb_c h_a^{\phantom{a}b}+ \delb_a h_c^{\phantom{c}b}- \delb^b h_{ac}) + O(h^3)\\
&= \half \check{e}_\mu^{\phantom{\mu}c}\check{T}^a_{\phantom{a}b}\delb_c h_a^{\phantom{a}b}+ O(h^3).
\end{align}
Finally, we use the field equations (\ref{FEqs}) to write
\begin{align}\label{delJ}
\phi^*(\nabla_a J_\mu^{\phantom{\mu}a})
&= \frac{1}{2\kappa}\check{e}_\mu^{\phantom{\mu}q}(\widehat{G}_{ab}^{\phantom{ab}cd}h_{cd}) (\delb_q h^{ab}) + O(h^3).
\end{align}
Inserting (\ref{delj}) and (\ref{delJ}) into (\ref{conservation}), and discarding the $O(h^3)$ terms, we arrive at the defining relation of the gravitational energy-momentum tensor:
\begin{align}\label{tdef}
\kappa \delb^a \tau_{aq}= -\half (\delb_q h^{ab}) \widehat{G}_{ab}^{\phantom{ab}cd}h_{cd}.
\end{align} 

The next step will be to use this equation to derive a formula for $\tau_{ab}$ in terms of $h_{ab}$. In order to do so, however, we must make one additional demand: $\tau_{ab}$ will not depend on second derivatives of $h_{ab}$, but will be a function of $\delb_c h_{ab}$ and $\gb_{ab}$ only. The reason we must impose this condition is that equation (\ref{tdef}) can only define $\tau_{ab}$ up to the addition of ``superpotential'' terms, those fields whose divergence vanishes \emph{identically}. Because these terms are of the form $\delb^c\delb^d H_{[ac][bd]}$ (where $H_{[ac][bd]}=H_{[bd][ac]}$ is some function of $h_{ab}$) they necessarily contain second derivatives; thus our restriction on $\tau_{ab}$ is sufficient to remove this ambiguity. At the moment, it might be tempting to view this condition as a convenient way to tame the derivation, and keep in mind that we can always add in super-potentials later if we wish. However, in section \ref{props} it will become clear that many of the interesting properties displayed by $\tau_{ab}$ will be unavoidably spoilt by the addition of such terms. For this reason we will not consider superpotentials further here.

\subsection{Determining the Energy-Momentum Tensor}\label{Determin}
In truth, it will not be possible to construct a symmetric tensor $\tau_{ab}$ that satisfies (\ref{tdef}) for all $h_{ab}$;\footnote{We will shortly describe how to check this assertion, which is simply a property of (\ref{tdef}) and independent of the requirement that $\tau_{ab}$ contain no  second derivatives.} to make progress we will need to impose some condition on $\delb_c h_{ab}$ and specialise to this restricted set of gravitational fields. Although this forced restriction might appear to be a flaw in our formalism, as we shall soon see, it is actually a valuable asset.

There are only three linear conditions we can place on $\delb_c h_{ab}$ which neither introduce extra fields, nor break Lorentz invariance: (a) $\delb_c h_{ab}=0$, (b) $\delb_a h=0$, or (c) $\delb^a h_{ab}=\lambda \delb_b h$, for some constant $\lambda$.\footnote{The only other possibility, $\delb^ah_{ab}=0$, can be achieved by taking (c) with $\lambda=0$.} Condition (a) is obviously far too restrictive: it does not allow us any gravitational field whatsoever. In contrast, condition (b) is not restrictive enough: there is no $\tau_{ab}$ that solves (\ref{tdef}) for all gravitational fields with constant trace.\footnote{For the sake of brevity, we will not prove this assertion here. Instead we will attend to condition (c) and derive the formula for $\tau_{ab}$ that it admits. After we have done so, we invite the reader to perform a similar calculation under condition (b) and verify that no solution exists.} We must therefore focus on condition (c), which we repeat for later reference: 
\begin{align}\label{lambdadef}
\delb^ah_{ab}=\lambda \delb_b h.
\end{align} 
Using this relation, it will be possible to replace any occurrence of $\delb^ah_{ab}$ with $\lambda\delb_b h$; hence the most general formula for a symmetric tensor $\tau_{ab}$, a quadratic function of $\delb_c h_{ab}$, is as follows:
\begin{widetext}
\begin{align}\nonumber
\kappa \tau_{pq}&= \gb_{pq}(A_0 \delb_c h_{ab}\delb^c h^{ab}+ A_1 \delb_a h  \delb^a h + A_2\delb_c h_{ab}\delb^b h^{ac} ) + A_3 \delb_p h_{ab}\delb_q h^{ab}  + A_4 \delb_p h\delb_q h + A_5 \delb_a h \delb_{(p}h_{q)}^{\phantom{q)}a} \\\label{ansatz}
&\quad + A_6 \delb^a h^b_{\phantom{b}(p}\delb_{q)} h_{ab} + A_7 \delb_a h_{bp}\delb^a h_{q}^{\phantom{q)}b} + A_8 \delb_b h_{ap}\delb^a h_{q}^{\phantom{q)}b}   + A_9\delb_a h  \delb^a h_{pq},
\end{align}
where $\{A_n\}$ are arbitrary constants. We proceed by substituting this ansatz into (\ref{tdef}) and solving for $\{A_n\}$. First, let us calculate $\delb^p\tau_{pq}$ by taking the divergence of (\ref{ansatz}); using (\ref{lambdadef}) to convert every $\delb^ah_{ab}$ to $\lambda \delb_b h$, and collecting terms, we find that left-hand side of (\ref{tdef}) amounts to
\begin{align}\nonumber
\kappa \delb^p\tau_{pq}&= (2 A_0 +A_3) \delb_q \delb_c h_{ab}\delb^c h^{ab} + (2 A_1 + A_4 + \half \lambda A_5 + \lambda A_9) \delb_a h \delb^a \delb_q h + (2 A_2 + \half A_6) \delb_c h_{ab}\delb_q \delb^b h^{ac}  \\\nonumber
&\quad + A_3 \delb^2 h_{ab} \delb_q h^{ab} + A_4 \delb^2 h \delb_q h + (\half A_5 + \lambda A_7 + \lambda A_8 +A_9) \delb_a \delb^b h \delb^a h_{bq} + (\half A_5 +\half \lambda A_6)\delb_a \delb_b h \delb_q h^{ab}\\\label{tderiv1}
&\quad + \half A_5 \delb^a h \delb^2 h_{aq}+ (\half A_6 + A_7) \delb^a \delb_b h_{cq} \delb_a h^{bc} +\half A_6 \delb_a h_{bq}\delb^2 h^{ab} +  A_8 \delb^c h_{ab}\delb^a \delb^b h_{cq}.
\end{align}
Meanwhile, (\ref{lambdadef}) simplifies the right-hand side of (\ref{tdef}):
\begin{align}\label{tderiv2}
-\half (\delb_q h^{ab}) \widehat{G}_{ab}^{\phantom{ab}cd}h_{cd}=-\half \delb_q h^{ab} \left((\lambda-\half)\delb_a\delb_b h - \half \delb^2 h_{ab} + \half \gb_{ab}(1-\lambda)\delb^2 h \right).
\end{align}
Comparing (\ref{tderiv1}) with (\ref{tderiv2}), term by term, we conclude that the \emph{unique} solution to (\ref{tdef}) is
\begin{align}\bs
A_0=-\tfrac{1}{8},\quad
A_1=\tfrac{1}{16},\quad
A_3 = \tfrac{1}{4},\quad
A_4 = -\tfrac{1}{8},\quad
A_2=A_5=A_6=A_7=A_8=A_9=0,\quad
\lambda = \half.\es
\end{align}
\end{widetext}
We have therefore determined the formula for our gravitational energy-momentum tensor,
\begin{align}\nonumber
\kappa \tau_{pq}&= \tfrac{1}{4} \delb_p h_{ab}\delb_q h^{ab}- \tfrac{1}{8} \delb_p h\delb_q h \\\label{tau}&\quad - \tfrac{1}{8}\gb_{pq}(\delb_c h_{ab}\delb^c h^{ab}- \half \delb_ah\delb^ah),
\end{align}
and the condition,
\begin{align}
\delb^a h_{ab} -\half \delb_b h =0,
\end{align}
that we must place on the gravitational field.  Finally, we introduce the abbreviation $\hb_{ab}=h_{ab} -\half \gb_{ab}h$ for the trace-reversed gravitational field, and $\bar{\tau}_{ab}=\tau_{ab} -\half \gb_{ab}\tau$ for the trace-reversed gravitational energy-momentum tensor, allowing us to compactly re-express our results:
\begin{align}
\label{taubar}
\kappa \bar{\tau}_{pq}=\tfrac{1}{4}\delb_p  h_{ab}\delb_q \hb^{ab},
\end{align}
\begin{align}\label{dedonder}
\delb^a \hb_{ab} =0.
\end{align}

We are now in a position to justify our earlier claim, that the restriction on the gravitational field is not a hindrance but a major advantage. The condition we have derived (\ref{dedonder}) is simply the defining equation of the \emph{harmonic gauge}.\footnote{This is also commonly referred to as de Donder gauge or Lorentz gauge.} As this equation can \emph{always} be satisfied by making a gauge transformation, it does not in any way limit the physical applicability of our approach! Moreover, we have received a valuable gift: this condition has appeared as a natural consequence of the derivation, and forces upon us a very strong restriction for the diffeomorphism $\phi$ that maps the physical spacetime onto the background. Specifically, $\phi^{-1}$ is required to map the Lorentz coordinates of the background onto \emph{harmonic} coordinates in the physical spacetime.\footnote{To see this, let $\{x^\mu\}$ be Lorentz coordinates on the background ($\delb_a \delb_b x^\mu=0$), and let $\{y^\mu\}$ be coordinates in physical spacetime defined by $y^\mu(p)= x^\mu (\phi(p))$ for all $p \in \mc{M}$. Then $\phi^* (\nabla^2 y^\mu)=  \delb^2 x^\mu -h^{ab} \delb_a \delb_b x^\mu - (\delb^a h_{ab} - \half \delb_b h) \delb^b x^\mu =0$.} We can therefore think of (\ref{dedonder}) as the condition that specifies the correct behaviour to demand of the basis $\{e_\mu^{\phantom{\mu}a}\}$ needed to define sensible energy-momentum currents $J_\mu^{\phantom{\mu}a}$. We should stress, however, that while the harmonic gauge condition has removed the vast majority of the gauge freedom, a small amount remains in the form of transformations $\delta h_{ab}=\delb_{(a} \xi_{b)}$ which satisfy $\delb^2 \xi_a=0$; we will return to this issue in section \ref{int}.

This completes the derivation of $\tau_{ab}$. We have found the unique symmetric tensor, a quadratic function of $\delb_c h_{ab}$, that describes the transfer of energy and momentum between matter and the gravitational field according to (\ref{conservation}). In doing so we have found that this solution only exists if the harmonic condition $\delb_a \hb^{ab} =0$ is obeyed, and  this in turn has solidified the definition of energy-momentum currents $J_\mu^{\phantom{\mu}a}$ as the contraction of $T_{ab}$ with the basis vectors associated with \emph{harmonic} coordinate systems of physical spacetime. In the next section we shall prove that $\tau_{ab}$ displays many other interesting properties very much in keeping with its interpretation as an energy-momentum tensor. In section \ref{int} we shall examine energy-momentum exchange in detail, and address the last piece of gauge freedom. 

\section{Properties}\label{props}
Here we will demonstrate that, in two important special cases, $\tau_{ab}$ exhibits interesting mathematical properties (beyond accounting for $\nabla_a J_\mu^{\phantom{\mu}a}$) that further promote its interpretation as the energy-momentum tensor of the linear gravitational field.

\subsection{Gauge Invariance of Plane-Waves}\label{GIPW}
Consider an arbitrary gravitational \emph{plane-wave}:
\begin{align}\label{plane}
h_{ab}=h_{ab}(k_{\mu}x^\mu).
\end{align}
Here, $k_{a}$ is a constant vector, and $\{x^\mu\}$ are Lorentz coordinates on the background. The linear vacuum field equation
$\delb^2 \bar{h}_{ab}=0$, and the harmonic condition $\delb^a \bar{h}_{ab}=0$, enforce
\begin{align}
k^ak_a=0,\quad k^a\bar{h}^\prime_{ab} =0,
\end{align}
respectively, where the prime indicates differentiation with respect to the variable $k_{\mu}x^\mu$. We wish to consider the most general gauge transformation $\delta h_{ab}=\delb_{(a}\xi_{b)}$ that maintains the plane-wave form of $h_{ab}$. Clearly we require $\xi_a = \xi_a(k_{\mu}x^\mu)$, and thus
\begin{align}\label{planeGT}
\delta h_{ab}= k_{(a} \xi^\prime_{b)}.
\end{align}
Note that $k^a k_a=0$ now guarantees $\delb^2 \xi_a=0$, ensuring that the harmonic condition (\ref{dedonder}) is not broken by the transformation. Let us now calculate the effect of this transformation on the gravitational energy-momentum tensor; working from (\ref{taubar}),
\begin{align}\nonumber
\kappa \delta \bar{\tau}_{pq}&= \half \delb_p \delta h_{ab} \delb_q \bar{h}^{ab} + \tfrac{1}{4} \delb_p \delta h_{ab}\delb_q \delta \bar{h}^{ab}\\\nonumber
&= \half k_pk_qk_{(a}\xi^{\prime\prime}_{b)}\bar{h}^{\prime ab} \\\nonumber&\quad+ \tfrac{1}{4} k_pk_qk_{(a}\xi^{\prime\prime}_{b)}(k^{(a}\xi^{\prime\prime b)} - \half \gb^{ab}k^c\xi^{\prime\prime}_c)\\
&= \tfrac{1}{8}k_pk_q ( (k^c\xi^{\prime\prime}_c)^2 - (k^c\xi^{\prime\prime}_c)^2  )=0.
\end{align}
Thus the energy-momentum of an arbitrary gravitational plane-wave is completely invariant under the gauge freedom consistent with the harmonic condition and its plane-wave form. This is significant for a number of reasons. Firstly it reveals that, for the special case of plane-waves, we need not concern ourselves with the gauge freedom that remains after enforcing the harmonic condition: the requirement that the gauge be chosen such that the plane-wave form of the field be manifest is sufficient to unambiguously define the energy-momentum tensor $\tau_{ab}$ from the physical spacetime $(\mc{M},g)$. Thus, even if one does not accept our method for resolving the last of the gauge ambiguity (to be presented in section \ref{int}) it is still possible to stop at this point and agree that a well-defined energy-momentum tensor for gravitational plane-waves has been found. Secondly, this particular gauge invariance will prove useful when we wish to produce a \emph{global} picture of the motion of energy-momentum: if the \emph{source region} of a gravitational wave is very far from the \emph{detection region}, we may use a different gauge in each and yet still produce a consistent picture of energy-momentum transfer -- the field in intermediate region will approximate a plane-wave, and thus $\tau_{ab}$ in this region will agree with both end-point gauges.\footnote{This idea is explained fully in section \ref{EMTT}.} There is also a third significance to this result, but this will only become apparent once we have demonstrated the second important property of gravitational energy-momentum: positivity.

\subsection{Positivity}\label{Positivity}
This section concerns the energy-momentum of \emph{transverse-traceless} (\TT) gravitational fields, those for which $h=0$, $\delb^ah_{ab}=0$, and $u^a h_{ab}=0$, for some constant timelike vector field $u^a$ defined on the background.\footnote{Clearly, the harmonic condition is satisfied as a result of these requirements.} For now let us simply suppose that these conditions apply to $h_{ab}$ and derive the consequences for $\tau_{ab}$. We shall justify our interest in this specialisation, and offer an interpretation of $u^a$, in section \ref{int}. Presently, let it suffice to say that because these conditions may always be imposed (at least locally) by a gauge transformation in regions where $\check{T}_{ab}=0$, the results of this section will be generally applicable to vacuum regions, but it will not be necessary to demand that $\check{T}_{ab}=0$ globally. We present our result as the following theorem.

\begin{Theorem}
If, at some point $p\in \bc{M}$, the gravitational field $h_{ab}$  obeys the transverse-traceless conditions
\begin{align}\label{TTcon}
\delb^a h_{ab}=0, \quad h=0, \quad u^a h_{ab}=0,
\end{align}
for some timelike vector $u^a$, then $\tau_{ab}$ satisfies the following inequalities
\begin{align}\label{pos1}
v^a \tau_{ab} v^b &\ge 0,\\
\label{pos2}
v^a \tau_{ac}\tau^{c}_{\phantom{c}b} v^b &\le 0,
\end{align}
at $p$, for any timelike vector $v^a$.
\end{Theorem}
\begin{proof} Without loss of generality we can set $u^a u_a=-1$ and $v^a v_a=-1$. Now, introduce two Lorentzian coordinate systems at $p$, the first $\{x^0,x^i\}$ ($i=1,2,3$) such that  $u^0=1$, $u^i=0$, and the second $\{x^{0^\prime},x^{i^\prime}\}$ ($i^\prime=1,2,3$) such that  $v^{0^\prime}=1$, $v^{i^\prime}=0$. The transverse-traceless conditions (\ref{TTcon}) reduce (\ref{taubar}) to
\begin{align}
\kappa\tau_{pq}=-\tfrac{1}{8}\gb_{pq}\delb_c h_{ab}\delb^c h^{ab} + \tfrac{1}{4}\delb_p h_{ab}\delb_q h^{ab},
\end{align}
and set $h_{0i}=h_{00}=0$. Using the primed basis to express the tensor indices of $\gb_{ab}$ and $\delb_a$, the unprimed basis for $h_{ab}$, and writing $\dot{h}_{ij} \equiv \partial_{{0^\prime}} h_{ij}$, we find that
\begin{align}\nonumber
\kappa v^a\tau_{ab}v^b &=\kappa \tau_{{0^\prime}{0^\prime}}\\\nonumber
&=\tfrac{1}{8}((\partial_{i^\prime} h_{ij})^2-(\dot{h}_{ij})^2) + \tfrac{1}{4} (\dot{h}_{ij})^2\\
&= \tfrac{1}{8}((\partial_{i^\prime} h_{ij})^2+(\dot{h}_{ij})^2) \ge0,
\end{align}
because $(\partial_{i^\prime} h_{ij})^2\equiv\sum_{i,j,i^\prime =1}^3 \partial_{i^\prime} h_{ij}\partial_{i^\prime} h_{ij}$ and $(\dot{h}_{ij})^2 \equiv \sum_{i,j=1}^3\partial_{0^\prime} h_{ij}\partial_{0^\prime} h_{ij}$  are sums of squares. Similarly,
\begin{align}\nonumber
(4\kappa v^a \tau_{ab})^2 &= 16\kappa^2  (-(\tau_{{0^\prime}{0^\prime}})^2+(\tau_{{0^\prime} {i^\prime}})^2)\\\nonumber
&=-\tfrac{1}{4}((\partial_{i^\prime} h_{ij})^2+(\dot{h}_{ij})^2)^2 \\\nonumber&\quad+ (\dot{h}_{ij} \partial_{i^\prime} h_{ij})(\dot{h}_{kl} \partial_{i^\prime} h_{kl})\\\nonumber
&=-\tfrac{1}{4}((\partial_{i^\prime} h_{ij})^2 -(\dot{h}_{ij})^2)^2 - (\dot{h}_{ij})^2(\partial_{i^\prime} h_{kl})^2 \\\nonumber&\quad+ (\dot{h}_{ij} \partial_{i^\prime} h_{ij})(\dot{h}_{kl} \partial_{i^\prime} h_{kl})\\\nonumber
&= -\tfrac{1}{4}((\partial_{i^\prime} h_{ij})^2-(\dot{h}_{ij})^2)^2\\&\quad -\half (\dot{h}_{ij} \partial_{i^\prime} h_{kl}-\dot{h}_{kl} \partial_{i^\prime} h_{ij})^2\le0.
\end{align}
\end{proof}

The inequalities we have just deduced are the gravitational version of the \emph{Dominant Energy Condition}: the first indicates that $\tau_{ab}$ only ever defines positive energy-densities; the second indicates that the flux of this energy can never be spacelike. Succinctly, they tell us that gravitational energy is positive and never flows faster than light. As the Dominant Energy Condition has always referred to \emph{matter}, we will avoid confusion if we resist subsuming (\ref{pos1}) and (\ref{pos2}) under this name; instead, when the gravitational energy-momentum tensor obeys these inequalities (for all timelike $v^a$) we shall simply say that it is \emph{positive}, and write $\tau_{ab}\ge0$ as a shorthand.

That $\tau_{ab}$ is positive for all transverse-traceless $h_{ab}$ is one of the major advantages our approach has over previous descriptions of gravitational energy-momentum \cite{Einstein,LL,BG,Moller,Mannheim}. Provided we work with a transverse-traceless field (with respect to some timelike vector $u^a$), which is always possible locally in a vacuum, $\tau_{ab}$ will always make good physical sense in that it will obey its own version of the Dominant Energy Condition. To some extent, this result supplies its own justification for choosing the \TT-gauge whenever possible; however, we will see in the next section that these conditions arise naturally by considering the gauge invariant transfer of energy-momentum onto point-sources. Furthermore, the significance of $u^a$ (in terms of energy-momentum transfer) will also be explored through these arguments. 

Before we move on, however, we take this opportunity to present an important corollary of the plane-wave gauge-invariance of section \ref{GIPW}. It is well known that there always exists exactly one gauge transformation of the form (\ref{planeGT}) that takes an arbitrary plane-wave (\ref{plane}) to one obeying the \TT-conditions \cite{HobsonGR}. Hence we can transform any gravitational plane-wave into transverse-traceless gauge without altering the energy-momentum tensor, at which point the positivity theorem ensures that $\tau_{ab}\ge0$. Thus, all gravitational plane-waves have positive energy-momentum tensors, even if they are not transverse-traceless.

\section{Interactions}\label{int}
In this section we apply our formula for the gravitational energy-momentum tensor to the interaction between gravity and an idealised matter distribution that we shall refer to as a \emph{point-source}. The reader will be familiar with the \emph{compact source}, an isolated body confined to a compact spatial region of radius $d$ much smaller than the wavelength $\lambda$ of the gravitational radiation it emits; point-sources are the limit of such systems as $d\to0$, entirely analogous to the infinitesimal dipoles of electromagnetism.\footnote{We derive the energy-momentum tensor and gravitational field of the point-source in appendix \ref{sources}.} Not only will this provide a useful example of the practical application of our approach, these considerations will finally allow us to rid ourselves of the last trace of gauge-dependence in our description.

From now on we will work almost exclusively in the flat background spacetime; as such it will generally be convenient to represent all tensors in some Lorentzian coordinate system $\{x^\mu\}$, and to drop the ``caron'' mark from $\check{T}_{\mu\nu}$. Thus, our formula for the gravitational energy-momentum tensor is written as
\begin{align}\label{tauflat}
\kappa \bar{\tau}_{\mu\nu}=\tfrac{1}{4}\partial_\mu h_{\alpha\beta}\partial_\nu \hb^{\alpha\beta},
\end{align}
the harmonic condition becomes
\begin{align}\label{harmonicflat}
\partial^\mu \bar{h}_{\mu\nu}=0,
\end{align}
and the linearised field equations (\ref{FEqs}) are
\begin{align}\label{cFEqs}
\partial^2 \bar{h}_{\mu\nu}=-2\kappa T_{\mu\nu}.
\end{align}
Also, it will be useful to separate $\{x^\mu\}$ into a time coordinate $t= x^0$, and spatial coordinates $\vec{x}=(x^1,x^2,x^3)$, define a radial coordinate $r \equiv |\vec{x}|$, and to use lower-case Roman indices $(i,j,k\ldots)$ to indicate spatial components. Typically, the coordinates will be implicitly chosen to coincide with the rest-frame of the system under consideration.

\subsection{Pulses and Point-Sources}
The core of our analysis will be to examine the most localised gravitational interaction possible: an infinitesimal point-source (at $\vec{x}=0$) met by an instantaneous \emph{pulse} plane-wave (propagating along  the $x^1$ direction, arriving at $\vec{x}=0$ at $t=t_0$).

As a result of calculations in appendix \ref{sources}, we know that a point-source has the following energy-momentum tensor:
\begin{align}\nonumber
T_{00}&= M \delta(\vec{x}) + \half I_{ij} \partial_i \partial_j \delta(\vec{x}),\\
T_{0i}&= \half(\dot{I}_{ij} +J_{ij})\partial_j \delta (\vec{x}),\\\nonumber
T_{ij}&= \half \ddot{I}_{ij} \delta (\vec{x}).
\end{align}
Here $M$ and $J_{ij}= J_{[ij]}$ are constants representing, respectively, the mass and angular momentum of the source, and $I_{ij}=I_{(ij)}(t)$ its (time dependent) quadrupole moment.\footnote{The reader should refer to appendix \ref{sources} for definitions of these quantities in terms of the infinitesimal limit of the compact source.} We do not intend to use this point-source as an actual source of gravitational radiation, but rather as a probe of the energy-momentum of the incident pulse. To this end, we are interested in the limit $M,I_{ij}, J_{ij}\to0$, allowing us to neglect the self-interaction of the source. As this procedure is entirely analogous to using a ``test-particle'' to probe the geometry of spacetime, we shall refer to the point-source as a \emph{test-source} in this limit.

The gravitational field will consist of two parts: $h_{\mu\nu}= h^\text{source}_{\mu\nu} + h^\text{wave}_{\mu\nu}$.  The first, due to the test-source, is given in appendix \ref{sources} by (\ref{hpoint}) and satisfies the inhomogeneous field equations
\begin{align}
\partial^2 \bar{h}^\text{source}_{\mu\nu}&=-2\kappa T_{\mu\nu}.
\end{align}
The latter is the incident pulse plane-wave,
\begin{align}
h^\text{wave}_{\mu\nu} = A_{\mu\nu} H(k_\alpha x^\alpha-t_0),
\end{align}
where $H$ is the Heaviside step function, $k_\mu = (1,-1,0,0)$ is a null vector in the $x^1$ direction, and $A_{\mu\nu}$ is a constant tensor satisfying $k^\mu \bar{A}_{\mu\nu}=0$, as demanded by the harmonic condition. Obviously, $h^\text{wave}_{\mu\nu}$ satisfies the homogeneous field equations: $\partial^2 h^\text{wave}_{\mu\nu}=0$. 

Let us now compute $\partial^\mu \tau_{\mu\nu}$, which, via (\ref{conservation}), quantifies the exchange of energy-momentum between the test-source and the gravitational wave. Starting from (\ref{tauflat}), we have
\begin{align}\nonumber
\kappa \partial^\mu\tau_{\mu\nu}&= \tfrac{1}{4} \partial^2 \bar{h}_{\alpha\beta} \partial_\nu h^{\alpha\beta}\\\nonumber
&= - \half\kappa  T_{\alpha\beta}\partial_\nu (A^{\alpha\beta} H(k_\sigma x^\sigma-t_0)) \\&\quad+ O((h^\text{source}_{\mu\nu})^2),
\end{align}
and we neglect terms of order $(h^\text{source}_{\mu\nu})^2$ compared to those of order $h^\text{wave}_{\mu\nu}h^\text{source}_{\mu\nu}$ in the limit $M,I_{ij}, J_{ij}\to0$.\footnote{There is a slight technical issue here. From (\ref{hpoint}), we can see that, as $r\to 0$, $h^\text{source}_{\mu\nu}\to \infty$; thus $h^\text{source}_{\mu\nu}$ inevitably becomes larger than $h^\text{wave}_{\mu\nu}$ at small enough distances. Strictly speaking, then, one should use a finite-size source (of radius $d$, say) when one takes $M,I,J\to0$ and neglects $O((h^\text{source}_{\mu\nu})^2)$. As we can choose $d$ to be as small as we like, however, we can always replace the finite source with an equivalent point-source after this limit has been taken.} Using $H^\prime=\delta$, the Dirac delta function, we arrive at
\begin{align}\nonumber
\partial^\mu\tau_{\mu\nu}&= -\half k_\nu \delta(k_\sigma x^\sigma -t_0) T_{\alpha\beta}A^{\alpha\beta}\\\nonumber
&= -\tfrac{1}{4}k_\nu \delta(k_\sigma x^\sigma-t_0) \\\nonumber&\quad
\times\Big(\ddot{I}_{ij}A_{ij} \delta(\vec{x}) - 2 (\dot{I}_{ij} + J_{ij})\partial_j \delta (\vec{x})A_{i0}\\\label{notGI}
&\quad\quad\quad 
+ (2M\delta(\vec{x}) + I_{ij}\partial_i\partial_j \delta(\vec{x}))A_{00}\Big).
\end{align}
This is the equation we sought. It determines the energy and momentum collected by our probe due to the incident pulse, and locates this transfer in spacetime. The key problem is that above relation is not, as it stands, gauge invariant; we address this issue the next section. 

\subsection{Gauge Invariance and Microaveraging}\label{Micro}
The incident wave possesses gauge freedom that neither breaks the harmonic condition nor spoils its pulse plane-wave form:
\begin{align}\label{pulseGT}
\delta h^\text{wave}_{\mu\nu} = \partial_{(\mu}\xi_{\nu)};\qquad \xi_\mu = E_\mu \Delta(k_\alpha x^\alpha-t_0),
\end{align}
where $E_\mu$ is any constant vector, and $\Delta^\prime = H$. The effect of this transformation is to alter $A_{\mu\nu}$ by $\delta A_{\mu\nu}= k_{(\mu}E_{\nu)}$, and although the transverse components do not change ($\delta A_{22}= \delta A_{23}=\delta A_{33}=0$) the right-hand side of (\ref{notGI}) is clearly not invariant. The beauty of working with an instantaneous interaction, however, is that we can average over the (infinitesimal) interaction region
\begin{gather}\nonumber
\lim_{\epsilon \to 0} \mc{B}_\epsilon(t_0),\\ \ \text{where} \quad \mc{B}_\epsilon(t_0) \equiv \{(t,\vec{x}):  |t - t_0| \le \epsilon,|\vec{x}| \le \epsilon\},
\end{gather}
without sacrificing the localised description of $\partial^\mu \tau_{\mu\nu}$. Let us call this operation a \emph{microaverage} (at $\vec{x}=0$, $t=t_0$) and denote it by $\langle\ldots\rangle_{t_0}$:
\begin{align}\label{microaverage}
\langle f\rangle_{t_0}\equiv\delta(\vec{x})\delta(t-t_0) \lim_{\epsilon \to 0} \int_{\mc{B}_\epsilon (t_0)} f \ud^4 x.
\end{align}
For the interaction we are analysing, the integral $\int_{\mc{B}_\epsilon (t_0)} \partial^\mu\tau_{\mu\nu} \ud^4 x$ captures the key physical content of $\partial^\mu \tau_{\mu\nu}$. To elaborate: the divergence theorem equates this integral with $\int_{\partial \mc{B}_\epsilon (t_0)} \tau_{\mu\nu} \ud^4 S^\mu$ which measures the mismatch between the net flux of gravitational energy-momentum entering through a spherical surface $\mc{S}_\epsilon \equiv \{\vec{x}: |\vec{x}| \le \epsilon\}$ barely larger the source, and the gravitational energy-momentum contained within $\mc{S}_\epsilon$ that is gained between the times $t=t_0 -\epsilon$ and $t=t_0 + \epsilon$. By the defining property (\ref{tdef}) of $\tau_{\mu\nu}$, this mismatch in gravitational energy-momentum precisely accounts for the energy-momentum absorbed by the source, which is what we wanted to know. The only information we have lost in taking the microaverage is the knowledge of precisely where \emph{within the test-source} the energy-momentum is being absorbed. As we have let the size of this probe shrink to zero, however, this is of little concern. 

The computational advantage of the microaverage is that the integration in (\ref{microaverage}) allows us to transfer derivatives off the delta-functions in (\ref{notGI}); for example,
\begin{align}\nonumber
\int_{\mc{B}_\epsilon (t_0)} \delta(k_\alpha &x^\alpha-t_0)\dot{I}_{ij}\partial_j \delta(\vec{x}) A_{i0} \ud^4 x \\\nonumber&= -\int_{\mc{B}_\epsilon (t_0)} \partial_j\delta(t-x^1-t_0) \dot{I}_{ij} \delta(\vec{x}) A_{i0} \ud^4 x  \\\nonumber
&= \int_{\mc{B}_\epsilon (t_0)} \dot{\delta}(t-x^1-t_0) \dot{I}_{i1} \delta(\vec{x}) A_{i0} \ud^4 x  \\\nonumber
&= - \int_{\mc{B}_\epsilon (t_0)} \delta(t-x^1-t_0) \ddot{I}_{i1}\delta(\vec{x}) A_{i0} \ud^4 x\\
&= - \ddot{I}_{i1}(t_0) A_{i0}.
\end{align}
Applying this technique to the whole of (\ref{notGI}) yields
\begin{align}\nonumber
\langle \partial^\mu\tau_{\mu\nu} \rangle_{t_0} &= -\tfrac{1}{4}k_\nu \delta(\vec{x}) \delta(t-t_0) \\\label{mic1}&\quad \times\left(\ddot{I}_{ij}A_{ij} + 2 \ddot{I}_{i1}A_{i0} + \ddot{I}_{11}A_{00} + 2 M A_{00} \right).
\end{align}
Finally, we unpack $k^\mu\bar{A}_{\mu\nu}=0$,
\begin{align}\label{unpack}\bs
\Rightarrow A_{00}+ A_{11} +2 A_{01} =0,\quad A_{22}+A_{33}=0,\\\quad A_{02} +A_{12}=0,\quad A_{03} +A_{13}=0,\es
\end{align}
and substitute these into (\ref{mic1}). The result is
\begin{align}\nonumber
\langle \partial^\mu\tau_{\mu\nu} \rangle_{t_0} &= -\half k_\nu \delta(\vec{x}) \delta(t-t_0) \\\label{maint}&\quad\times \left(\ddot{I}_{\times}A_{\times} + \ddot{I}_{+}A_{+} +M A_{00} \right),
\end{align}
where we have written the transverse components of the wave as $A_\times =  A_{23}$ and $A_+ = (A_{22}-A_{33})/2$, and extended this notation to $I_{ij}$. We are almost done: $\delta A_\times=\delta A_+=0$ under the gauge transformation (\ref{pulseGT}), so the first two terms in (\ref{maint}) are manifestly gauge invariant; however, the term proportional to $M A_{00}$ is not.

Various arguments can be made to show that this ``monopole term'' is physically irrelevant to the energy-momentum transfer we are considering. At the simplest level, the fact that we are free to set $A_{00}$ to any value (including zero) through gauge transformation (leaving $A_\times$ and $A_+$ untouched) indicates that the monopole term can have no bearing on the energy-momentum of the physical system under scrutiny. Furthermore, if we consider a wave for which $A_{\mu\nu}=k_{(\mu} E_{\nu)}$, then it is clear that, while such a pulse is gauge-equivalent to flat spacetime ($A_{\mu\nu}=0$) it would nonetheless register a transfer of energy-momentum if the monopole term were to be believed.

The physical irrelevance of the monopole term should come as no great surprise, as there can be no way to extract energy-momentum from a gravitational wave using a monopole alone (i.e.\ a test-source with $I_{ij}=J_{ij}=0$): an observer sitting on an isolated point mass could perform no local test to distinguish whether a gravitational wave had even passed, and in particular, must be unable to extract any energy.

In fact, all that the monopole term is responding to is a change in normalisation of the time coordinate in the physical spacetime: $\phi^*(e^a_0 e_{0a}) = -1 +h_{00}$. Naively, we might expect this factor to be significant as it represents the Newtonian potential at the test-source. However, this is not a local effect. The only way an observer on the test-source could be aware of such a shift is by comparison with some standard clocks at spatial infinity. The pulse plane-wave prevents this idea from being well-defined, however, as it divides spatial infinity into two regions: $x^1 < t-t_0 $, where wave has already been received, and $x^1 > t-t_0$, where it has not. Fortunately, a gauge can always be chosen that does not suffer from this inconsistency; setting $A_{00}=0$ is the only way to ensure that the standard clocks at infinity all run at the \emph{same} rate (relative to our coordinate $t$) and this inevitably removes all trace of the monopole term from the interaction. Thus the insistence that the clocks at infinity agree with each other amounts to a prescription that removes the gauge-dependence of our microaveraged energy-momentum transfer.\footnote{The reader should not be under the impression that the monopole term is \emph{universally} insignificant. Thus far we have argued its irrelevance only for pulse plane-wave, and as we shall see at the beginning of the next section, this idea follows by linearity to general gravitational waves. However, should the gravitational field have a \emph{time-independent} part, then it is possible for this to couple to the monopole in a physically meaningful way. This is due to the particularly limited gauge freedom available to $h_{00}$ when the field is time-invariant. We will return to this issue in Section \ref{tindep}.} We can implement this procedure mathematically (without fixing the gauge, or setting $M=0$, which is physically untenable) by acting on $\langle \partial^\mu\tau_{\mu\nu} \rangle_{t_0}$ with the operator $(1-M\partial_M)$:
\begin{align}\nonumber
\langle \partial^\mu\tau_{\mu\nu} \rangle_{t_0}^\slashed{M} &\equiv (1-M\partial_M)\langle \partial^\mu\tau_{\mu\nu} \rangle_{t_0}\\\label{monopolefree}& = -\half k_\nu \delta(\vec{x}) \delta(t-t_0) \left(\ddot{I}_\times A_{\times} + \ddot{I}_+ A_{+}  \right).
\end{align}
We shall call this the \emph{monopole-free} microaverage. This is a local, completely gauge invariant description of the energy-momentum transferred onto test-sources by pulse plane-waves. Furthermore, the right-hand side of (\ref{monopolefree}) has an obvious physical interpretation: the coupling between $\ddot{I}_{ij}$ and $A_{ij}$ can be understood, roughly speaking, as the product of a force (responsible for accelerating  the constituents of the quadrupole moment) and a distance (actually an expansion/contraction of spacetime) and thus represents the work done on the test-source. For example, consider a test-source composed of two bodies of mass $m$ separated by a light elastic rod of length $2d$ aligned with the $x^2$-axis; provided the amplitude of the motion of the masses is much smaller than $d$, then $\ddot{I}_{22}\cong 4d m a$, where $a$ is the (outward) acceleration of each mass. Due to the gravitational wave, the proper distance of each mass from the centre of the rod increases by $ A_{22}d/2$; thus, counting the motion of both ends of the rod, the total work done on the source, by the wave, is $-2 (ma) (A_{22}d/2) = -\ddot{I}_{22} A_{22}/4= - \ddot{I}_+ A_+ /2$, which agrees precisely with (\ref{monopolefree}). Thus we see that the monopole-free microaverage corresponds to the familiar physical quantities that we would intuitively use to define the energy and momentum of the test-source.

In the next section we will generalise the monopole-free microaverage to arbitrary gravitational fields, and uncover a substantial mathematical shortcut that will greatly simplify this procedure.

\subsection{Arbitrary Gravitational Fields}\label{Arb}
Clearly, pulse waves are a special case, and one might expect that for an arbitrary (harmonic gauge) plane-wave
\begin{align}\label{Bplane}
h^\text{wave}_{\mu\nu}= B_{\mu\nu}(k_\alpha x^\alpha),\qquad k^\mu\bar{B}_{\mu\nu}=0,
\end{align}
we would need to perform \emph{finite} averages, rather than microaverages, to remove the gauge dependence of our description; thus, we would be forced to sacrifice our \emph{localised} picture of energy-momentum transfer. However, provided $B_{\mu\nu}(t)\to 0$ as $t\to-\infty$, we can always write
\begin{align}\nonumber
h^\text{wave}_{\mu\nu}&= \int^{\infty}_{-\infty} B_{\mu\nu}(t_0)\delta(k_\alpha x^\alpha-t_0)\ud t_0\\\nonumber
&= \int^{\infty}_{-\infty} \dot{B}_{\mu\nu}(t_0)H(k_\alpha x^\alpha-t_0)\ud t_0 \\\nonumber&\quad- \left[ B_{\mu\nu}(t_0)H(k_\alpha x^\alpha-t_0)\right]^{+\infty}_{-\infty}\\
&= \int^{\infty}_{-\infty} \dot{B}_{\mu\nu}(t_0)H(k_\alpha x^\alpha-t_0)\ud t_0,
\end{align}
and perform the monopole-free microaverage on each component of this sum:
\begin{align}\nonumber
&\left\langle\partial^\mu \tau_{\mu\nu}[h^\text{source}_{\alpha\beta}+h^\text{wave}_{\alpha\beta}]\right\rangle_{\int}^\slashed{M} \\\label{intmicro}&\equiv \int^{\infty}_{-\infty}\langle \partial^\mu\tau_{\mu\nu}[h_{\alpha\beta}^\text{source}+ \dot{B}_{\alpha\beta}(t_0) H(k_\sigma x^\sigma-t_0)] \rangle_{t_0}^\slashed{M}  \ud t_0.
\end{align}
The result of this process is
\begin{align}\label{Btrans}
\left\langle\partial^\mu \tau_{\mu\nu}\right\rangle^\slashed{M}_{\int}= -\half k_\nu \delta(\vec{x}) \left(\ddot{I}_{\times}\dot{B}_{\times} + \ddot{I}_+ \dot{B}_{+}\right),
\end{align}
which renders the interaction completely gauge invariant, and does not sacrifice the local character of our description of energy-momentum transfer.\footnote{Just as we write $I_{ij}$ for $I_{ij}(t)$, we have, in (\ref{Btrans}) and elsewhere, left the argument of $B_{\mu\nu}(t)$ implicit.}

Although the operation of splitting the wave into a series of pulses and performing a monopole-free microaverage on each pulse may seem too complicated to be useful, the same result can be achieved by a simple alternative method: transform $h_{\mu\nu}^\text{wave}$ to transverse-traceless gauge, where the vector $u^\mu$ referred to by the \TT-conditions (\ref{TTcon}) corresponds to the rest-frame of the test-source. Then, when we calculate $\partial^\mu \tau_{\mu\nu}$, we will automatically recover the monopole-free microaveraged result. To demonstrate this, we recalculate $\partial^\mu \tau_{\mu\nu}$, generalising (\ref{notGI}) for use with arbitrary plane-waves (\ref{Bplane}),
\begin{align}\nonumber
\partial^\mu\tau_{\mu\nu}&= -\tfrac{1}{4}k_\nu \left(\ddot{I}_{ij}\dot{B}_{ij} \delta(\vec{x}) - 2 (\dot{I}_{ij} + J_{ij})\partial_j \delta (\vec{x})\dot{B}_{i0}\right.\\&\qquad\qquad  \ \left.
 {}+ (2M\delta(\vec{x}) + I_{ij}\partial_i\partial_j \delta(\vec{x}))\dot{B}_{00}\right),
\end{align}
and substitute the \TT-conditions $B_{0\nu}=B=0$ (which, along with $k^\mu\bar{B}_{\mu\nu}=0$, set $B_{1\nu}=0$ and $B_{22}=-B_{33}$):
\begin{align}\nonumber
\partial^{\mu}\tau^{\TT}_{\mu\nu}& \equiv  \partial^{\mu}\tau_{\mu\nu}[h^\text{source}+(h^\text{wave})^\TT] \\\nonumber&= -\tfrac{1}{4}k_{\nu} \delta(\vec{x})\ddot{I}_{ij}\dot{B}_{ij} \\&=- \half k_\nu \delta(\vec{x}) \left(\ddot{I}_{\times}\dot{B}_{\times} + \ddot{I}_{+}\dot{B}_{+} \right).
\end{align}
Hence,
\begin{align}\label{MA=TT}
\left\langle\partial^\mu \tau_{\mu\nu}\right\rangle^\slashed{M}_{\int} = \partial^{\mu}\tau^{\TT}_{\mu\nu}.
\end{align}
Furthermore, this equation is not only applicable to incident plane-waves. Because both sides are linear in $h^\text{wave}_{\mu\nu}$, equations (\ref{MA=TT}) must also hold when $h^\text{wave}_{\mu\nu}$ is any \emph{sum} of plane-waves, propagating in arbitrary directions. Locally, we can always express $h^\text{wave}_{\mu\nu}$ as a sum of plane-waves (and some time-independent part, which we will ignore until section \ref{tindep}) so, quite generally, we have
\begin{align}\label{MA=TT2}
\left\langle\partial^\mu \tau_{\mu\nu}\right\rangle^\slashed{M}_{\int} = -\tfrac{1}{4} \delta(\vec{x})\ddot{I}_{ij}\partial_\nu h^\TT_{ij},
\end{align}
where $h^\TT_{\mu\nu}$ is the incident gravitational field in transverse-traceless gauge. This equation provides an easy method for calculating the energy-momentum transferred onto the microaveraged test-source due to the presence of arbitrary incident gravitational radiation. Moreover, we see that (ignoring the time-independent field) our gauge invariant probe only exchanges energy-momentum with the transverse-traceless field; the other components of the field do not play a role in this process. We explore the wider significance of this result in the next section.

\subsection{Energy-Momentum and Transverse-Traceless Gauge}\label{EMTT}
Let us now take a step back from the fine details of the test-source interaction and assess the general picture that is unfolding. As we first saw in section \ref{Determin}, $\tau_{\mu\nu}$ is not in general invariant under the gauge freedom that remains after the harmonic condition has been enforced. As a partial remedy of this, the monopole-free microaveraged test-source emerged as a local, gauge invariant probe of gravitational energy-momentum exchange. It has now come to light that only the transverse-traceless field takes part in this process. From this standpoint, a method suggests itself which will remove the remaining ambiguity of $\tau_{\mu\nu}$ in a natural fashion: simply transform the incident gravitational field to transverse-traceless gauge! Consequently, only the degrees of freedom relevant to gauge invariant energy-momentum exchange will contribute to the gravitational energy-momentum tensor. We shall codify this idea as a gauge-fixing program defined in terms of two ``frames'', one associated with gravitational detectors, the other with astrophysical sources. In what follows, the gauge-fixing only refers to the \emph{dynamical part} of the gravitational field; as we explain in section \ref{tindep}, the time-independent part of the gravitational field is essentially gauge invariant and so does not need to be fixed in any way.

\emph{Detector-frame}. Consider a gravitational detector $D$ in a region $\mc{V}_D$ which contains no matter besides the detector. We shall suppose that the \emph{incident field} (due to sources outside $\mc{V}_D$) is much larger than the field due to the detector itself; in other words, we model $D$ as a test-source. The detector-frame is then obtained by transforming the \emph{incident field} to \TT-gauge, taking $u^\mu$ to be the four-velocity of the detector.\footnote{Note that it is not the total field, but just the incident field, which is made transverse-traceless. We cannot alter the gauge of the field generated by the detector because it is impossible to produce outgoing spherical waves in the gauge field $\xi^\mu$ without breaking the harmonic condition at $D$. This cannot be allowed to happen if $\tau_{\mu\nu}$ is to account for the energy-momentum exchanged with the detector.} As a result, the energy-momentum transferred onto $D$ will be exactly equal to the gauge invariant quantities defined by the monopole-free microaverage. What is more, we can imagine adding hypothetical test-sources (co-moving with $D$) anywhere within $\mc{V}_D$ in order to ``measure'' the gravitational energy-momentum there; because the field has been prepared in this gauge, the result will agree with the gauge invariants we have already defined. In this way, the detector-frame defines $\tau_{\mu\nu}$ through the gauge invariant energy-momentum that would be absorbed by furnishing $\mc{V}_D$ with an array of infinitesimal probes moving at the same velocity as the actual detector.

\emph{Source-frame}. Now consider a compact source $S$ in a region $\mc{V}_S$ which, as above, contains no other matter. In contrast to the detector, we shall assume any incident field can be neglected in comparison to the \emph{outgoing field} due to $S$. The source-frame is obtained by transforming the \emph{outgoing field} to \TT-gauge, taking $u^\mu$ to be the four-velocity of the source. This gauge transformation can only be achieved by breaking the harmonic condition at $S$ (see Appendix \ref{TTgauge} for details) so it will not be possible to use $\tau_{\mu\nu}$ to describe the energy-momentum lost by the source; however, this self-interaction is ill-defined for a point-like system anyway, and we would have to resolve the source into component parts before such a question could be answered. Outside the source, $h_{\mu\nu}$ will remain harmonic, so $\tau_{\mu\nu}$ will still represent the energy-momentum that could be absorbed by a hypothetical test-source (with the velocity of $S$) were we to insert one. Much like the detector-frame, we can think of this prescription as measuring $\tau_{\mu\nu}$ by filling $\mc{V}_S$ with infinitesimal probes that are co-moving with the source.

The gauge-fixing program is simple: if one wishes to describe the energy-momentum of the gravitational field as it would be measured by some detector $D$, then adopt the source-frame near distant astrophysical sources, and the detector-frame of $D$ everywhere else. This allows the energy-momentum in the vicinity of a source to be unambiguously determined by the source alone, whilst simultaneously adapting the gravitational field for a description of energy-momentum absorption by the detector. Remarkably, despite using \emph{different} \TT-gauges in the various regions, this program still produces a self-consistent picture of the propagation of energy-momentum from the sources to the detector. This is because, many wavelengths from an isolated source, the gravitational field approximates a plane-wave, and so (as we saw in section \ref{GIPW}) the gravitational energy-momentum is gauge invariant there. Thus, in this regime, the source-frame energy-momentum and the detector-frame energy-momentum \emph{are equal}, regardless of the relative velocity between the source and the detector. The plane-wave regions will therefore ``sew together'' the different frames and produce a globally consistent description of energy-momentum flowing from source to detector. When the sources are not isolated, but are separated from each other by only a small number of wavelengths, no plane-wave region will exist between the sources. In this case, the sources must be treated as one extended source, with a joint source-frame that identifies $u^\mu$ with the four-velocity of the centre of mass of the many-body system. We illustrate the gauge-fixing program schematically in figure \ref{Sewing}.

\begin{figure}
\centering
\includegraphics[scale=.45]{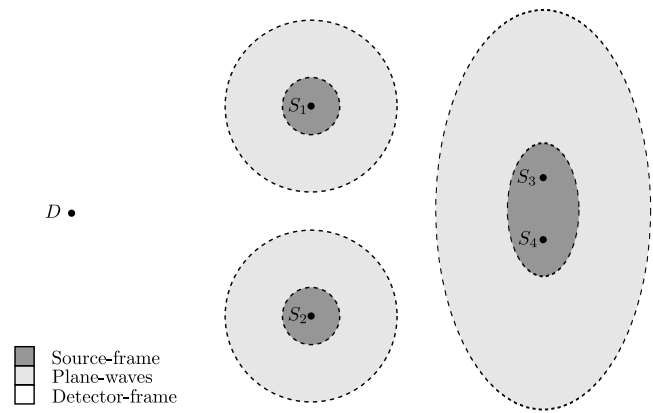}
\caption{Schematic showing plane-wave regions sewing various source-frames to the detector-frame of $D$. $S_1$ and $S_2$ represent compact sources many wavelengths apart; if they have different velocities, then each will determine a separate source-frame. In contrast, $S_3$ and $S_4$ are only separated by a small number of wavelengths; as there is no plane-wave region dividing the sources, they must share their source-frame. It is natural to base this joint source-frame on the velocity of the centre of mass of the multi-component system.}
\label{Sewing}
\end{figure}

Up to this point, we have justified our insistence on \TT-gauge based purely on considerations of energy-momentum exchange with matter. Of course, there is another exceptional property of $\tau_{\mu\nu}$, derived in section \ref{Positivity}, that also holds under these conditions: it is always \emph{positive}. This is a peculiar and surprising result. It is something of a small miracle that transverse-tracelessness guarantees not only agreement with the monopole-free microaverage, but also ensures that $\tau_{\mu\nu}$ represents positive energy-density, and causal energy flux. Under these conditions it will always be possible to make physical sense of the gravitational energy-momentum tensor: we will never have to interpret (or explain) negative or superluminal energy. 

We feel that the dual significance of transverse-traceless gauge leaves little doubt that this is the correct procedure by which to remove the final trace of ambiguity in the definition of $\tau_{\mu\nu}$. In section \ref{Aps} we will apply this program to a small number of examples, including two specific compact sources: a vibrating rod, and an equal-mass binary. First, however, we must address a technical issue regarding the time-independent part of the gravitational field. 

\subsection{Time-Independent Fields}\label{tindep}
Unlike the dynamical part of the gravitational field, the gauge of the time-independent mode (i.e. the time-averaged field) is completely fixed by insisting that (a) the harmonic condition $\partial^\mu\bar{h}_{\mu\nu}=0$ holds everywhere, (b) $h_{\mu\nu}\to 0$ as $r\to \infty$, and (c) gauge transformations $\delta h_{\mu\nu}=\partial_{(\mu}\xi_{\nu)}$ are finite everywhere and bounded at infinity. To see this, suppose that our time-independent field $h_{\mu\nu}(\vec{x})$ obeys the harmonic condition and vanishes at spatial infinity. Then the transformed field $h^\prime_{\mu\nu}= h_{\mu\nu}+ \partial_{(\mu}\xi_{\nu)}$ will only satisfy the harmonic condition if $\partial^2 \xi=0$, and will only be time-independent if $\ddot{\xi}_\mu=0$. Thus $\partial_i^2 (\partial_{(\mu}\xi_{\nu)}) =0$, the only bounded solutions of which are constants (by Liouville's theorem). Hence we are forced to take $h^\prime_{\mu\nu}= h_{\mu\nu}$ if the new field is to also vanish at spatial infinity, and we thus conclude the gauge is unique.\footnote{This argument relies on our insistence that the harmonic condition be valid \emph{everywhere}. In section \ref{TTsource} we will make use of the following mathematical trick: by relaxing the harmonic condition at the source itself, we will be able to combine many local \TT-gauges to form a gauge in which the dynamical part of the gravitational field (outside the source) is transverse-traceless for all time. As we will see, however, it is impossible to apply this procedure to a time-independent field. Thus there is nothing to gain from weakening the harmonic condition on the time-independent mode, and it is therefore kept unbroken.}

This result reveals that we are not required to perform any form of microaverage to remove the gauge dependence of the energy-momentum transfer associated with the time-independent mode of the gravitational field: this mode is already gauge invariant. In truth, this is a rather convenient situation. We could not microaverage a time-independent field even if we needed to, due to the caveat $B_{\mu\nu}(t)\to 0$ as $t\to-\infty$, encountered when deriving (\ref{intmicro}).

As there is no gauge freedom in the time-independent part of the gravitational field, we cannot expect this mode to have $h_{0\mu}=0$ or $h=0$ in general. This leaves open the possibility (at least in principal) that close to sources, where the time-independent mode can become comparable in amplitude to the dynamical field, the positivity of $\tau_{\mu\nu}$ may be compromised. However, in section \ref{TTsource} we will see that, even very close to (but not inside) a compact source, the time-independent field obeys $\bar{h}_{00}\gg \bar{h}_{0i}\gg \bar{h}_{ij}$.\footnote{These order-of-magnitude inequalities are not limited to the compact source. The time-independent mode of the gravitational field is always generated by the time-averaged energy-momentum tensor of matter $\langle T_{\mu\nu}\rangle$, and it is to be expected that this field will be dominated by the slow (i.e.\ non-relativistic) motion of matter, so that $\langle T_{00}\rangle \gg \langle T_{0i}\rangle\gg\langle T_{ij}\rangle$. Hence $\bar{h}_{00}\gg \bar{h}_{0i}\gg \bar{h}_{ij}$ will hold quite generally: whenever $\langle T_{\mu\nu}\rangle$ is dominated by non-relativistic motion.} It is easy to show that such a field will not upset the positivity of $\tau_{\mu\nu}$. Neglecting the small quantities, the trace-reversed gravitational field will take the form
\begin{align}
\bar{h}_{\mu\nu}=h^\TT_{\mu\nu} - 4 \Phi u_\mu u_\nu,
\end{align}
where $h^\TT_{\mu\nu}$ is the dynamical field in transverse-traceless gauge, and $\Phi \equiv - \bar{h}_{00}/4$ is the the Newtonian potential, the only non-negligible contribution from the time-independent mode. The trace-reversed gravitational energy-momentum tensor therefore takes the form
\begin{align}\label{N}
4\kappa \bar{\tau}_{\mu\nu}=\partial_\mu h_{\alpha\beta}\partial_\nu \hb^{\alpha\beta} = \partial_\mu h^\TT_{ij}\partial_\nu h^\TT_{ij} + 8\partial_\mu \Phi  \partial_\nu \Phi,
\end{align}
which ensures that the positivity proof of section \ref{Positivity} can proceed almost exactly as before, with $\Phi$ effectively behaving as an additional component of $h^\TT_{ij}$. Thus, even though it is not transverse-traceless, the time-independent mode does not give rise to any negative or superluminal energy.

\section{Applications}\label{Aps}
This section is devoted to calculating the energy-momentum content of the gravitational field in a small number of examples, following the gauge-fixing program of section \ref{EMTT}.

\subsection{Plane-Waves}
Although we have already studied gravitational plane-waves in a variety of contexts, we have yet to evaluate the energy-momentum they carry. This calculation will serve as a simple first example, and will illustrate the use of the detector-frame. 

We begin with an arbitrary (harmonic gauge) plane-wave,
\begin{align}
h_{\mu\nu}=h_{\mu\nu}(k_{\alpha}x^\alpha),\quad
k^\mu k_\mu=0,\quad k^\mu\bar{h}^\prime_{\mu\nu} =0,
\end{align}
and substitute this field into equation (\ref{tauflat}):
\begin{align}\label{plane1}
\kappa \tau_{\mu\nu}= \tfrac{1}{4} k_\mu k_\nu h^\prime_{\alpha\beta}\bar{h}^{\prime\alpha\beta}.
\end{align}
As we ascertained in section \ref{GIPW}, the energy-momentum of  plane-waves is gauge invariant. Consequently, we can simplify (\ref{plane1}) by evaluating $h_{\mu\nu}$ in \TT-gauge, thereby removing all components except for $h_+$ and $h_\times$:
\begin{align}\label{plane2}
\kappa \tau_{\mu\nu}= \half k_\mu k_\nu ( (h^\prime_+)^2 + (h^\prime_\times)^2).
\end{align}
In fact, because $\delta h_+= \delta h_\times=0$ under any gauge transformation that keeps $h_{\mu\nu}$ a plane-wave, the right-hand side of this equation is gauge invariant also. Hence, equation (\ref{plane2}) must hold in any gauge, and all other terms on the right-hand side of (\ref{plane1}) must cancel in general.\footnote{This can be verified by taking $k_\mu = (1,-1,0,0)$ and using $k^\mu \bar{h}_{\mu\nu}=0$ in the same form as (\ref{unpack}).} Using this formula for $\tau_{\mu\nu}$, every future-directed timelike unit-vector $v^\mu$ defines a gravitational energy current-density,
\begin{align}
v^\mu\tau_{\mu\nu} = v^\mu k_\mu k_\nu ( (h^\prime_+)^2 + (h^\prime_\times)^2)/2\kappa,
\end{align}
which is clearly future-directed and null; unsurprisingly, the energy of a gravitational plane-wave is positive and flows at the speed of light in the direction of propagation.

So far, the gauge invariance of $\tau_{\mu\nu}$ has made gauge-fixing unnecessary. The insistence that we evaluate the $h_{\mu\nu}$ in the detector-frame only becomes important when there are multiple plane-waves propagating in different directions. Suppose, for example, that there are two plane-waves: 
\begin{align}\label{2planes}
h_{\mu\nu}=h^\text{I}_{\mu\nu}(k^\text{I}_\alpha x^\alpha)+ h^\text{II}_{\mu\nu}(k^\text{II}_\alpha x^\alpha), 
\end{align}
where $k^\text{I}_\mu$ and $k^\text{II}_\mu$ are non-parallel null vectors. As $\tau_{\mu\nu}$ is quadratic in $h_{\mu\nu}$, the energy-momentum of the total field takes the form
\begin{align}\nonumber
\kappa \tau_{\mu\nu}&= \kappa \tau^\text{I}_{\mu\nu} + \kappa \tau^\text{II}_{\mu\nu} + \half k^\text{I}_{(\mu}k^\text{II}_{\nu)}h^{\text{I}\prime}_{\alpha\beta}\bar{h}^{\text{II}\prime\alpha\beta} \\\label{2planes2}&\quad- \tfrac{1}{4}\eta _{\mu\nu}k^\text{I}_{\sigma}k^{\text{II}\sigma} h^{\text{I}\prime}_{\alpha\beta}\bar{h}^{\text{II}\prime\alpha\beta},
\end{align}
where $\tau^\text{I}_{\mu\nu}$ and $\tau^\text{II}_{\mu\nu}$ are the individual energy-momentum tensors of $h^\text{I}_{\mu\nu}$ and $h^\text{II}_{\mu\nu}$ respectively. Now, any gauge transformation that preserves the form (\ref{2planes}) of the gravitational field can be thought of as a pair of gauge-transformations that act on $h^\text{I}_{\mu\nu}$ and $h^\text{II}_{\mu\nu}$ separately, preserving their plane-wave forms; thus $\tau_{\mu\nu}^\text{I}$ and $\tau_{\mu\nu}^\text{II}$ must be invariant under gauge-transformations of this type. However, the ``cross-terms'' in (\ref{2planes2}) \emph{are} gauge-dependent, as the (gauge-dependent) longitudinal components of $h^\text{I}_{\mu\nu}$ will be transverse to $h^\text{II}_{\mu\nu}$, and vice versa. This gauge ambiguity is removed, however, by the presence of a physical detector: once we demand that the energy-momentum exchanged with this detector is to equal the monopole-free microaverage, we fix the gauge completely. This is the detector-frame: $h_{\mu\nu}$ is transverse-traceless, with $u^\mu$ identified as the four-velocity of the detector. In this sense, the gauge-fixing program is the procedure that enables us to ``add together'' the energy-momentum tensors of gravitational plane-waves (which, individually, are gauge invariant) to form the energy-momentum tensor of the total field.

For the sake of the concreteness, let us set $k^\text{I}_\mu=(1,1,0,0)$, $k^\text{II}_\mu=(1,0,1,0)$, and $u^\mu=(1,0,0,0)$. Then, once we have transformed $h_{\mu\nu}$ to \TT-gauge, the energy-momentum tensor becomes
\begin{align}\nonumber
\kappa \tau_{\mu\nu}&=  \kappa \tau^\text{I}_{\mu\nu} + \kappa \tau^\text{II}_{\mu\nu} - \tfrac{1}{4}(2 k^\text{I}_{(\mu}k^\text{II}_{\nu)} +  \eta_{\mu\nu} )h^{\text{I}\prime}_+ h^{\text{II}\prime}_+,
\end{align}
where $h^\text{I}_+= h^\text{I}_{22}=-h^\text{I}_{33}$, and  $h^\text{II}_+ = h^\text{II}_{33}=-h^\text{II}_{11}$. Due to the positivity theorem of section \ref{Positivity}, we already know this tensor describes a positive energy-density, and a causal energy flux. As a particular example of this, it is easy to calculate the energy-density associated with $u^\mu$,
\begin{align}\nonumber
\kappa\tau_{00}&= \half\left((h^{\text{I}\prime}_+)^2 +(h^{\text{I}\prime}_\times)^2 +(h^{\text{II}\prime}_+)^2 +(h^{\text{II}\prime}_\times)^2\right) \\&\quad- \tfrac{1}{4}h^{\text{I}\prime}_+h^{\text{II}\prime}_+,
\end{align}
and, as $(h^{\text{I}\prime}_+)^2 + (h^{\text{II}\prime}_+)^2 \ge h^{\text{I}\prime}_+h^{\text{II}\prime}_+/2$, we can confirm that this energy-density can never be negative.

\subsection{Linearised Schwarzschild Spacetime}\label{LinSch}
The Schwarzschild spacetime is the vacuum solution to the Einstein field equations outside any uncharged spherical non-rotating body of mass $M$. At distances much greater than $\kappa M$, where the linear approximation is valid, the gravitational field must  therefore correspond to that of the compact source with $I_{ij}=J_{ij}=0$:
\begin{align}\label{Sch}
\bar{h}_{00}=\frac{\kappa M}{2 \pi r},
\end{align}
and $\bar{h}_{0i}=\bar{h}_{ij}=0$. Obviously, this is an example of a gravitational field that is entirely time-independent; thus, as explained in section \ref{tindep}, there will be no possibility of transforming to transverse-traceless gauge, nor any need to do so.\footnote{Of course, the linearised Schwarzschild spacetime can be represented in other gauges, but no others obey $\partial^\mu \bar{h}_{\mu\nu}=0$ and $\dot{h}_{\mu\nu}=0$ everywhere, and are well-behaved at infinity.} The formula (\ref{tauflat}) for the gravitational energy-momentum tensor yields
\begin{align}\label{Mtau}
\tau_{\mu\nu} = \kappa \left(\frac{M}{8\pi r^2}\right)^2 \left(2\hat{x}_\mu \hat{x}_\nu - \eta_{\mu\nu}\right),
\end{align} 
where $\hat{x}^\mu$ is the radial unit vector.\footnote{This is a trivial extension of the notation $\hat{x}_i=x_i /r$ from appendix \ref{sources}; we simply define $\hat{x}_0=0$.} It is easy to confirm that this energy-momentum tensor is everywhere positive. Any timelike unit vector $v^\mu$ defines a positive gravitational energy-density,
\begin{align}
\varrho \equiv v^\mu\tau_{\mu\nu} v^\nu  = \kappa \left(\frac{M}{8\pi r^2}\right)^2 \left(2(\hat{x}_i v_i)^2 +1 \right)\ge0,
\end{align}
and an energy current-density $J^\nu \equiv v^\mu \tau_{\mu\nu}$ which is nowhere spacelike:
\begin{align}
J^\nu J_\nu = - \kappa^2  \left(\frac{M}{8\pi r^2}\right)^4\le 0.
\end{align}
It is also worth comparing equation (\ref{Mtau}) with the electromagnetic energy-momentum outside a point-charge: $T_{\mu\nu}\sim (g_{\mu\nu}+2u_{\mu}u_{\nu}- 2\hat{x}_\mu \hat{x}_\nu)/r^4$. While both tensors diminish in proportion to $1/r^4$, they define very different stress profiles at each point. The gravitational field has $\tau_{rr}=-\tau_{\theta\theta}=-\tau_{ii} =\tau_{00}\ge0$ and thus describes radial compression, tangential tension, and negative pressure; while the electromagnetic field has $-T_{rr}= T_{\theta\theta}=T_{ii}=T_{00}\ge0$ and thus describes radial tension, tangential compression, and positive pressure. The physical significance of this difference is far from obvious, but may relate to the like-attracts-like character of gravity: the negative gravitational pressure mediating the attraction of other masses, while positive electromagnetic pressure causes the repulsion of like-charges. In addition, it may be possible to understand the radial gravitational compression (and tangential tension) in terms of some ``elastic'' analogy for spacetime, as the Schwarzschild geometry ``squeezes in'' extra radial distance (between spheres of given area) in comparison to flat space. However, the theoretical value of such an analogy is unclear, and we do not intend to develop it any further here.

Although we have focused here on the linearised Schwarzschild spacetime as a particular example of a time-independent field, we note in passing that it is easy to evaluate the gravitational energy-momentum tensor associated with \emph{any} static configuration of matter $T_{\mu\nu}=u_\mu u_\nu \rho(\vec{x})$: equation (\ref{tauflat}) simplifies to 
\begin{align}\label{newton}
\kappa\tau_{\mu\nu}= 2 \partial_\mu \Phi \partial_\nu \Phi - \eta_{\mu\nu}\partial_\alpha \Phi \partial^\alpha \Phi,
\end{align}
where the Newtonian potential $\Phi\equiv -\bar{h}_{00}/4$  is determined by solving Poisson's equation $\partial_i^2 \Phi = \kappa \rho/2$. Equation (\ref{newton}) reveals that the energy-momentum of the Newtonian potential is exactly that of a massless Klein-Gordon scalar field.

\subsection{Gravitational Field of a Compact Source}\label{TTsource}
We shall now calculate the energy-momentum content of the gravitational field (\ref{hpoint}) generated by a compact source.\footnote{This calculation should not be confused with the analysis performed in section \ref{int}, where a test-source (essentially a compact source in the limit $d, M, J_{ij}, I_{ij}\to0$) interacted with an incident field, which presumably had been generated by another source, very far way. Here the compact source  will represent an astrophysical source (with finite $d$, $M$, $J_{ij}$ and $I_{ij}$) and by adopting the source-frame we will compute the energy-momentum of the outgoing field as it would be measured by microaveraged detectors co-moving with the source.} The first step will be to enter the source-frame: we must transform the dynamical part of the outgoing field into \TT-gauge, with $u^\mu$ identified as the four-velocity of the source. We can always make this transformation \emph{locally} by choosing the gauge fields $\xi_{\mu}$ such that $h_{0\mu}=\dot{h}_{0\mu}=0$ and $h=\dot{h}=0$ at some time $t=t_0$; then $\partial^2 \xi_\mu=0$ (which preserves the harmonic condition) and the field equations $\partial^2 h_{\mu\nu}=0$ (outside the source) ensure that $h_{0\mu}=0$ and $h=0$ continues to be true for $t\in (t_0 - r, t_0 + r)$.\footnote{See \citep[chap. 4.4b]{Wald} for details.} This method is problematic in that it is based around an arbitrary special time $t_0$, and that transverse-tracelessness always breaks down within a time $\Delta t = 2r$; these issues prevent us from forming a \emph{global} picture of the energy-momentum outside the source.

As we show in appendix \ref{TTgauge}, these problems can be completely avoided if we weaken the harmonic condition slightly, so that $\partial^\mu \bar{h}_{\mu\nu}=0$ is only enforced \emph{outside the source}. This trick allows us to find a gauge in which the dynamical field is transverse-traceless everywhere outside the source, for all $t$, and does not require us to choose a special time $t_0$. We can think of this gauge as a way of joining up the many possible local gauges (defined using the aforementioned method) in a mutually consistent fashion.\footnote{Presumably, there is some topological obstruction which prevents us from joining these local gauges without violating the harmonic condition at the source. However, provided we do not intend to calculate the energy-momentum transferred between matter and gravity at the source, this is not an issue. Even if we were careful to keep the gravitational field harmonic at the source, more work would be needed to perform such a calculation, as this self-interaction only becomes well-defined by breaking down the compact source into component parts.} The process of transforming the gravitational field of the compact-source (\ref{hpoint}) can be found in appendix \ref{TTgauge}, here we simply display the result:
\begin{align}\label{TThpoint}\bs
\bar{h}_{00}&= \left(2M + \langle I_{ij} \rangle\partial_i\partial_j\right)\frac{\kappa}{4 \pi r}\\
 \bar{h}_{0i} &= J_{ij}\partial_j \frac{\kappa}{4 \pi r}\\
\bar{h}_{ij} &= \int^\infty_{-\infty} \frac{\ud \omega}{2\pi} e^{i\omega t} \Bigg[\left( \tilde{I}_{kl}\delta_{ij} + \tilde{I} \delta_{ik}\delta_{jl} - 4\tilde{I}_{k(i}\delta_{j)l} \right)\partial_k \partial_l  \\
&\quad \ \! \qquad \qquad \qquad- \frac{1}{\omega^2} \tilde{I}_{kl}\partial_k \partial_l \partial_i\partial_j  
\\
&\quad \ \! \qquad \qquad \qquad + \omega^2 \left( \tilde{I} \delta_{ij}-  2\tilde{I}_{ij}\right) \Bigg] \frac{\kappa e^{-i\omega r}}{8\pi r},
\es\end{align}
where we have introduced the notation
\begin{align}
\langle I_{ij} \rangle \equiv \lim_{\Delta \to \infty} \int^\infty_{-\infty} I_{ij}(t) \frac{ e^{-t^2/\Delta^2}}{\sqrt{ \pi \Delta^2}} \ud t,
\end{align}
for the time-average of the quadrupole moment, and
\begin{align}
\tilde{I}_{ij}(\omega)\equiv \int^\infty_{-\infty} e^{- i \omega t}\left(I_{ij}(t)-\langle I_{ij} \rangle \right) \ud t,
\end{align} 
for the Fourier transform of its dynamical part. Notice that the terms proportional to $M$, $J_{ij}$, and $\langle I_{ij}\rangle$ constitute the time-independent mode of the field, and have therefore not been transformed. At this point we can confirm the assertion of section \ref{tindep}, that the time independent field satisfies $\bar{h}_{00}\gg \bar{h}_{0i}\gg \bar{h}_{ij}$. To do so we note that, firstly, there is no time-independent term in $\bar{h}_{ij}$, and secondly, seeing as the radius of the source $d\gtrsim J/M$, and that we are outside the source (which is to say, $r\gg d$, the regime of validity of (\ref{TThpoint})) then we must have $J/r \ll M$.

Having rendered the dynamical field transverse-traceless outside the source, all that remains is to substitute (\ref{TThpoint}) into (\ref{tauflat}) to calculate $\tau_{\mu\nu}$. As was shown in the process of deriving (\ref{N}), the energy-momentum of the time-independent field adds \emph{linearly} (i.e.\ without cross-terms) to that of the dynamical field. Given that we have already investigated the part due to the time-independent field in section \ref{LinSch}, it is generally more interesting to discard this term, and focus on the additional energy-momentum due to the dynamical field. In figure \ref{Frames} we show the results of a computation of this additional gravitational energy-density $\tau_{00}$ outside two monochromatic compact sources: a vibrating rod, and an equal-mass binary.  It goes without saying that the energy-density is everywhere positive, and that the energy current-density is nowhere spacelike.

\begin{figure*}
\centering
\subfigure{\includegraphics[scale=.52]{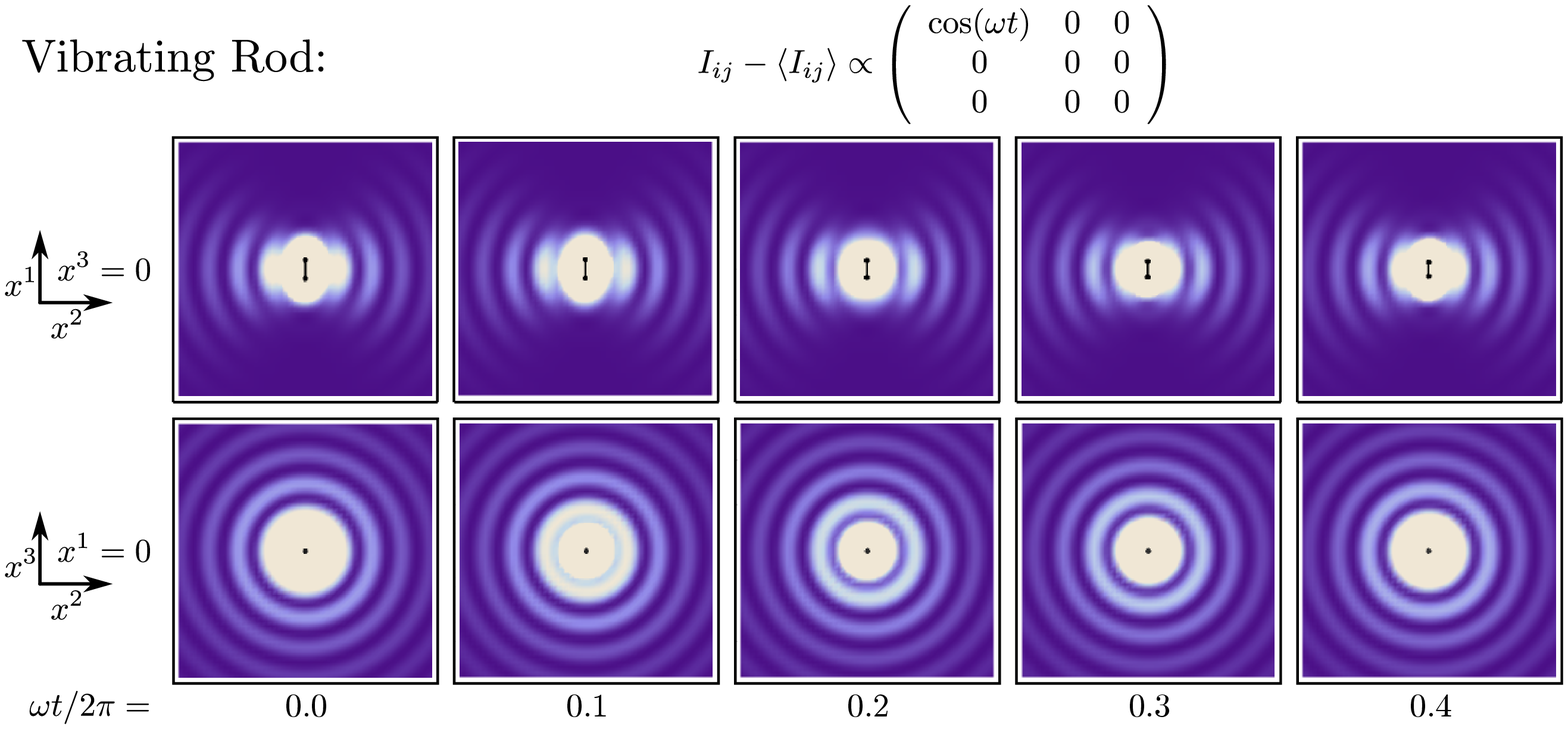}}\vspace{.1in}
\subfigure{\includegraphics[scale=.52]{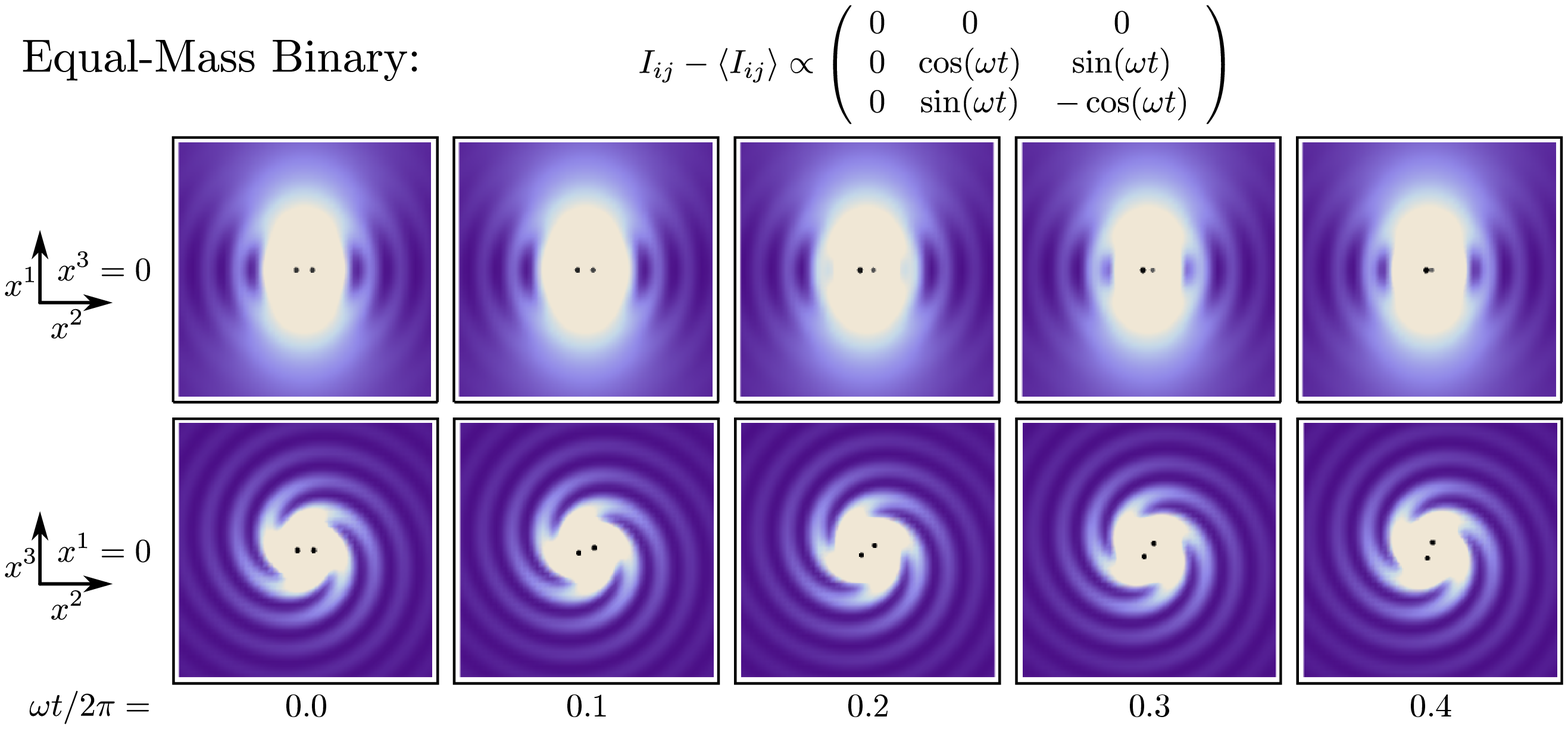}}
\caption{Plots of the energy-density of the dynamical gravitational field outside two monochromatic compact sources: a vibrating rod, and an equal-mass binary. Only half a period is shown, as $\tau_{\mu\nu}$ oscillates with twice the frequency of the source. Although the rod and the binary are much smaller than one wavelength, they have been magnified to illustrate the phase of their motion. The  propagation of gravitational energy is more easily appreciated in the animated versions of these plots, available at \url{www.mrao.cam.ac.uk/~lmb62/animations} .}
\label{Frames}
\end{figure*}

\section{Conclusion}
It is natural to suspect that wherever matter gains energy under the influence of gravity, a corresponding loss in the energy of the gravitational field must have occurred. By constructing a framework to quantify this idea, we have succeeded in localising the energy and momentum of the linear gravitational field, and have shown  this energy to be positive and to not flow faster than light.

The core result of our investigation is the formula (\ref{taubar}) for the gravitational energy-momentum tensor, the \emph{unique} symmetric tensor, quadratic in $\delb_c h_{ab}$, which accounts for the energy-momentum lost or gained by matter through its interaction with gravity (\ref{conservation}). Crucially, a tensor satisfying these conditions only exists in the harmonic gauge (\ref{dedonder}) and thus, as a matter of necessity rather than choice, our framework discards nearly all its gauge freedom. A small set of viable gauge transformations still remain, however, and although these do not alter the energy-momentum of gravitational plane-waves (\S\ref{GIPW}) this invariance does not extend to arbitrary gravitational fields. 

In response to this ambiguity, the monopole-free microaverage  was developed (\S\ref{Micro}); this constitutes a local and fully gauge-invariant description of energy-momentum transfer, and agrees with the intuitive notion that the ``work done'' on a gravitational detector is the product of the force (proper acceleration) and the proper distance through which the force is applied. Of the incident field, only the transverse-traceless part contributes to the microaveraged exchange (\ref{MA=TT}), and thus a natural gauge-fixing program is motivated, based around transverse-traceless gauge (\S \ref{EMTT}). The effect of this program is to prepare the field so that no microaverage is needed, and furthermore, to ensure that energy-momentum is only assigned to those components of the field whose energy-momentum can be measured by a microaveraged detector. Because the positivity property (\S \ref{Positivity}) holds true wherever the field is transverse-traceless, the gauge-fixing procedure also ensures that (for the dynamical field at least) gravitational energy-density is positive, and gravitational energy flux is timelike or null. No longer burdened by gauge ambiguity, the gravitational energy-momentum tensor can be evaluated without difficulty: the energy-momentum of gravitational plane-waves (\ref{plane2}), the linearised Schwarzschild spacetime (\ref{Mtau}), and the gravitational radiation outside compact sources (Fig.\ \ref{Frames}) have been provided as specific examples.

With regards to further investigation, there are two obvious directions in which our framework might be extended: beyond the linear approximation, and beyond the flat background.\footnote{To extend $\tau_{ab}$ beyond the linear regime, one would hope to construct a tensor $t_{ab}$, defined on the physical spacetime $\mc{M}$, such that $\phi^* t_{ab} = \tau_{ab} + O(h^3)$. Clearly, it will only be possible to make this identification if $\tau_{ab}$ is gauge invariant to second order, as $\delta (\phi^* t_{ab})\sim  O(h^2) \partial \xi $ under a change of gauge, whereas $\delta\tau_{ab} \sim  O(h) \partial \xi $ unless it is invariant. Thus, only once $\tau_{ab}$ has been brought into the detector-frame, or the source-frame, can we proceed. Consequently, we should expect that $t_{ab}$ will not only depend on the physical metric $g_{ab}$, but also on the four-velocity of the relevant detector or source.} However, it is currently unknown whether such extensions are possible, or even conceptually sound. On a more practical level, one could apply our formalism to the energetics of actual gravitational detectors, rather than the idealised test-sources that have so far dominated our discussion. In doing so, the framework developed here may benefit the design and analysis of future gravitational-wave experiments.

\begin{acknowledgments}
L.\ M.\ B.\ is supported by STFC and St.\ John's College, Cambridge.
\end{acknowledgments}

\appendix

\section{Sources}\label{sources}
The aim of this Appendix is to derive the formula for $T_{\mu\nu}$ that defines a gravitational point-source (essentially an infinitesimal gravitational quadrupole) and the field $h_{\mu\nu}$ that it generates. The derivation comprises two parts: first, a calculation of the field due to a compact source; second, a calculation of the field due to a candidate $T_{\mu\nu}$ that  vanishes everywhere but at $\vec{x}=0$. As the field from the first calculation matches that of the second (within the region of validity of the compact source approximation) we will be able to conclude that our candidate $T_{\mu\nu}$ is indeed the energy-momentum tensor we sought, that of an infinitesimal compact source.

\subsection{The Compact Source}\label{compact}
A \emph{compact source} is an isolated gravitational body confined to a compact spatial region $\mc{D}$ of radius $d$ much smaller than the wavelength $\lambda$ of the gravitational radiation it emits. Although calculations of the field $h_{\mu\nu}(\vec{x},t)$ outside a compact source are available in many standard references, we present our own here for two reasons. Firstly, textbook treatments commonly conflate the slow-motion approximation ($d\ll \lambda$) with the far-field approximation ($|\vec{x}|\equiv r \gg \lambda$). Here we shall assume only that the source is very small ($d\ll r,\lambda$) but not anything about the ratio of $\lambda$ to $r$.\footnote{As we are working with in the confines of linearised gravity, we should also insist that $d\gg 2\kappa M$, the Schwarzschild radius of the source. However, this will have little bearing on our calculation.} Secondly, the standard approaches frequently omit a full calculation of $\bar{h}_{00}$ and $\bar{h}_{0i}$. Presumably, these components are ignored because they do not appear to contribute to the gravitational field in transverse-traceless gauge; however, they must be included if $h_{\mu\nu}$ is to satisfy the harmonic condition.

The retarded solution to the linearised field equations (\ref{cFEqs}) is given by
\begin{align}\label{greens}
\bar{h}_{\mu\nu}(\vec{x},t) = \frac{\kappa}{2\pi} \int_{\mc{D}} \frac{T_{\mu\nu}(\vec{x}^\prime, t - |\vec{x}-\vec{x}^\prime|)}{|\vec{x}-\vec{x}^\prime|} \ud^3 x^\prime.
\end{align}
We shall proceed by expanding the right-hand side of this equation to second order in the small quantities $d/\lambda$ and $d/r$, so that we have an integral of energy-momentum tensors $T_{\mu\nu}\equiv T_{\mu\nu}({\vec{x}^\prime},t-r)$ evaluated at the \emph{same time} $t^\prime = t-r$. Using
\begin{align}
|\vec{x}-\vec{x}^\prime| &= r\left(1- \frac{\vec{x}\cdot\vec{x}^\prime}{r^2} + \frac{|\vec{x}^\prime|^2}{2r^2} - \frac{(\vec{x}\cdot\vec{x}^\prime)^2}{2r^4} + O((d/r)^3) \right),
\end{align}
equation (\ref{greens}) expands to
\begin{align}\nonumber
\bar{h}_{\mu\nu} &= \frac{\kappa}{2\pi r} \int_{\mc{D}} \! \ud^3 x^\prime \left[ T_{\mu\nu}\left(1+  \frac{\vec{x}\cdot\vec{x}^\prime}{r^2}- \frac{|\vec{x}^\prime|^2}{2r^2} + \frac{3(\vec{x}\cdot\vec{x}^\prime)^2}{2r^4}\right) \right.\\\nonumber
&\quad \qquad \qquad \qquad + r \dot{T}_{\mu\nu}\left( \frac{\vec{x}\cdot\vec{x}^\prime}{r^2} - \frac{|\vec{x}^\prime|^2}{2r^2}  + \frac{3(\vec{x}\cdot\vec{x}^\prime)^2}{2r^4}\right) \\\label{greenex}
&\quad \qquad \qquad \qquad \left.{} + r^2 \ddot{T}_{\mu\nu} \frac{(\vec{x}\cdot\vec{x}^\prime)^2}{2r^4} + O((d/r)^3)  \right].
\end{align}
Although we have not written their arguments, it should be understood that the $T_{\mu\nu}$ terms in the integral are evaluated at $(\vec{x}^\prime,t-r)$, while $\bar{h}_{\mu\nu}$ is evaluated at $(\vec{x},t)$.

In order to relate this integral to the basic physical properties of the source, we define its mass, momentum, and dipole moment by
\begin{align}\bs
M &\equiv \int_\mc{D} T_{00}\ud^3 x^\prime,\\
P_i &\equiv-\int_\mc{D} T_{0i}\ud^3 x^\prime,\\
X_i &\equiv \int_\mc{D} T_{00}  x_i^\prime\ud^3 x^\prime,\es
\end{align}
respectively. Notice that, because the source is entirely contained within $\mc{D}$ (so $T_{\mu\nu}=0$ on the boundary $\partial \mc{D}$) the conservation equation $\partial^\mu T_{\mu\nu}=0$ (the linearised version of (\ref{consT})) leads to the following relations:
\begin{align}\nonumber
\dot{X}_i &= \int_\mc{D}  \partial_0 T_{00}  x^\prime_i \ud^3 x^\prime = \int_\mc{D}  (\partial^\prime_j T_{j0} ) x^\prime_i \ud^3 x^\prime\\&= -\int_\mc{D} T_{j0} (\partial^\prime_j x^\prime_i) \ud^3 x^\prime = P_i,\\
\dot{P}_i &= -\int_\mc{D} \partial^\prime_j T_{ji}\ud^3 x^\prime =0.
\end{align}
Thus $\ddot{X}_i=0$, and we are free to fix $X_i=P_i=0$ by our choice of coordinate system. Note also that $\dot{M}=0$ follows by an identical argument. Next we define the quadrupole moment
\begin{align}
I_{ij} &\equiv \int_\mc{D} T_{00}x^\prime_i x^\prime_j\ud^3 x^\prime,
\end{align}
and then derive
\begin{align}\label{Idot}
\dot{I}_{ij} &= -2 \int_\mc{D}   T_{0(i}x_{j)}^\prime\ud^3 x^\prime,\\\label{Iddot}
\ddot{I}_{ij} &= 2 \int_\mc{D}  T_{ij}\ud^3 x^\prime,
\end{align}
in a similar fashion. Finally we define the angular momentum of the source
\begin{align}
J_{ij} \equiv - 2 \int_\mc{D}   T_{0[i}x_{j]}^\prime\ud^3 x^\prime,
\end{align}
and note that conservation sets $\dot{J}_{ij}=0$. 

Before substituting these definitions and results into (\ref{greenex}), note that equations (\ref{Idot}) and (\ref{Iddot}) indicate that $\int T_{0i}\ud^3 x^\prime \sim \dot{I}/d \sim M d/\lambda$ and $\int T_{ij}\ud^3 x^\prime \sim \ddot{I} \sim M d^2/\lambda^2$; hence the integrals of $T_{0j}$ and $T_{ij}$ already have (respectively) one and two extra factors of $(d/\lambda)$ than the integrals of $T_{00}$. Thus, to second order, $\bar{h}_{ij}$ will include contributions from only the zeroth order quantities multiplying $T_{ij}$ in (\ref{greenex}), and $\bar{h}_{0i}$ will include only first and zeroth order quantities multiplying $T_{0i}$. The final result, accurate to second order in the small quantities $(d/\lambda)$ and $(d/r)$, is therefore
\begin{align}\label{hpoint}
\bs\bar{h}_{00}&= \frac{\kappa}{4\pi}\Bigg(\frac{2M +\ddot{I}_{ij}\hat{x}_i\hat{x}_j}{r} + \frac{3\dot{I}_{ij}\hat{x}_i\hat{x}_j - \dot{I}}{r^2}\\&\quad\quad\quad\ \ + \frac{3I_{ij}\hat{x}_i\hat{x}_j -I}{r^3}\Bigg),\\
\bar{h}_{0i}&=-\frac{\kappa}{4\pi} \left( \frac{\ddot{I}_{ij}\hat{x}_j}{r} + \frac{\dot{I}_{ij}\hat{x}_j}{r^2}  +\frac{J_{ij}\hat{x}_j}{r^2} \right),\\
\bar{h}_{ij}&= \frac{\kappa \ddot{I}_{ij}}{4 \pi r},\es
\end{align}
where $\hat{x}_i= x_i/r$ is the radial unit vector, and all the $I_{ij}$ terms are evaluated at the retarded time $t^\prime= t-r$.
Note that, while the fields $\bar{h}_{00}$ and $\bar{h}_{0i}$ are often omitted from standard calculations, even in the far-field limit ($r\to\infty$), they still contain terms of equal size to $\bar{h}_{ij}$; these are necessary for consistency with the harmonic condition.

We have successfully derived the form of the gravitational field outside a compact source. However, because (\ref{hpoint}) was constructed under the approximation scheme $d\ll r$,  we can only trust these equations at distances much larger than the size of the source. However, we can still ask the following question: what source would produce a field such that (\ref{hpoint}) was valid for all $r$, no matter how small? This is the \emph{point-source} we have been interested in: the limit of the compact source as $d\to 0$. In the next section we present a candidate for the point-source, calculate its gravitational field, and show that this agrees with (\ref{hpoint}) for all $r$.

\subsection{The Point-Source}\label{point}
Consider the following energy-momentum tensor for matter:
\begin{align}\nonumber
T_{00}&= M \delta(\vec{x}) + \half I_{ij} \partial_i \partial_j \delta(\vec{x}),\\\label{Tpoint}
T_{0i}&= \half(\dot{I}_{ij} +J_{ij})\partial_j \delta (\vec{x}),\\\nonumber
T_{ij}&=  \half \ddot{I}_{ij} \delta (\vec{x}),
\end{align}
where $M$, $J_{ij}= J_{[ij]}$ are constants, $I_{ij}=I_{(ij)}(t)$ is independent of $\vec{x}$, and overdots indicate differentiation with respect to $t$. It is easy to check that this distribution obeys $\partial^\mu T_{\mu\nu}=0$. 

We wish to solve the linearised field equations
\begin{align}
\partial^2 \bar{h}_{\mu\nu}&=-2\kappa T_{\mu\nu},
\end{align} 
looking for the retarded solution. Recalling that
\begin{align}\label{fdelta}
\partial^2 \left(f(t-r)/r\right) = - 4 \pi \delta(\vec{x})f(t),
\end{align}
for any twice differentiable function $f(t)$, we see that we can replace $f \rightarrow \kappa \ddot{I}_{ij}/4\pi$ to generate the result
\begin{align}
\bar{h}_{ij}= \frac{\kappa \ddot{I}_{ij}(t-r)}{4 \pi r}.
\end{align}
Also, from equation (\ref{fdelta}), we have
\begin{align}
\partial^2 (\partial_j(f(t-r)/r)) = - 4 \pi f(t) \partial_j \delta(\vec{x}).
\end{align}
Thus, setting $ f \rightarrow \kappa(\dot{I}_{ij}+ J_{ij})/4\pi$
gives
\begin{align}\nonumber
\bar{h}_{0i}&=\frac{\kappa}{4\pi} \partial_j \left(\frac{\dot{I}_{ij}(t-r)+ J_{ij}}{r}\right)\\
&= -\frac{\kappa}{4\pi} \left( \frac{\ddot{I}_{ij}\hat{x}_j}{r} + \frac{\dot{I}_{ij}\hat{x}_j}{r^2}  +\frac{J_{ij}\hat{x}_j}{r^2} \right).
\end{align}
By the same method,
\begin{align}\nonumber
\bar {h}_{00}&= \frac{\kappa M}{2 \pi r} + \frac{\kappa }{4 \pi} \partial_i \partial_j \left(\frac{I_{ij}(t-r)}{r}\right)\\\nonumber
&= \frac{\kappa}{4\pi}\Bigg(\frac{2M +\ddot{I}_{ij}\hat{x}_i\hat{x}_j}{r} + \frac{3\dot{I}_{ij}\hat{x}_i\hat{x}_j - \dot{I}}{r^2}\\&\quad\quad\quad\ \ + \frac{3I_{ij}\hat{x}_i\hat{x}_j -I}{r^3}\Bigg).
\end{align}
Therefore the source (\ref{Tpoint}) generates a gravitational field identical to that of the compact source (\ref{hpoint}), except that these equations are now valid for all $\vec{x}$ (except, possibly, $\vec{x}=0$) not just $r \gg d$. The energy-momentum tensor (\ref{Tpoint}) is the point-source we required and (\ref{hpoint}) the field it generates; the correspondence with the compact source allows us to validate the interpretation of $M$ as the mass, $I_{ij}$ the quadrupole moment, and $J_{ij}$ the angular momentum of the source.

\section{Persistent Transverse-Traceless Gauge}\label{TTgauge}
Here we describe a method by which the dynamical part of the gravitational field outside a compact source (centred at $\vec{x}=0$) may be transformed to a gauge which remains transverse-traceless for all time, everywhere outside the source. This will be achieved by relaxing the harmonic condition slightly, so that $\partial^\mu \bar{h}_{\mu\nu}=0$ only holds outside the source. 

First, a point of notation. The gauge transformation described in this section is only applicable to the dynamical part of the gravitational field $h^\text{dyn}_{\mu\nu}\equiv h_{\mu\nu}- \langle h_{\mu\nu} \rangle$, where $\langle \ldots \rangle$ signifies a time average. Rather than crowd the notation, it will be convenient to assume that $\langle h_{\mu\nu} \rangle =0$, and use $h_{\mu\nu}$ to stand for $h^\text{dyn}_{\mu\nu}$. For the compact source, this amounts to setting $M=J_{ij}=\langle I_{ij}\rangle =0$ in (\ref{hpoint}). At the end of the calculation we will reinsert these time-independent terms to the transformed field without alteration.

The general procedure is as follows. To begin, take the Fourier transform of the dynamical part of the gravitational field:
\begin{align}\label{FT}
\tilde{h}_{\mu\nu}(\omega,\vec{x})\equiv \int^\infty_{-\infty} e^{- i \omega t}h_{\mu\nu}(t,\vec{x})\ud t.
\end{align}
The Fourier transform renders the field equations as
\begin{align}\label{FTFE}
(\omega^2 + \partial_i^2)\tilde{h}_{\mu\nu}(\omega,\vec{x})=0,
\end{align}
everywhere outside the source, i.e.\ for $\vec{x}\ne 0$. The harmonic condition becomes
\begin{align}\label{FTHC}
-i \omega \tilde{\bar{h}}_{0\mu} + \partial_i \tilde{\bar{h}}_{i\mu}=0,
\end{align}
and writing $\tilde{\xi}_{\mu}(\omega,\vec{x})$ for the Fourier transform of $\xi_{\mu}(t,\vec{x})$, the general gauge transformation $\delta h_{\mu\nu}= \partial_{(\mu}\xi_{\nu)}$ takes the form
\begin{align}\nonumber
\delta\tilde{h}_{00}&= i \omega \tilde{\xi}_0,\\
\delta\tilde{h}_{0i}&= \half i \omega \tilde{\xi}_i + \half \partial_i \tilde{\xi}_0,\\\nonumber
\delta\tilde{h}_{ij}&= \partial_{(i}\tilde{\xi}_{j)}.
\end{align} 
To achieve transverse-tracelessness we set
\begin{align}\bs
\tilde{\xi}_0 &= i\omega^{-1} \tilde{h}_{00},\\\label{gaugefix}
\tilde{\xi}_i &= 2 i \omega^{-1}\tilde{h}_{0i} - \omega^{-2} \partial_i \tilde{h}_{00}.\es
\end{align}
From the field equations (\ref{FTFE}) it is clear that this gauge transformation obeys $(\omega^2 +\partial_i^2)\tilde{\xi}_{\mu}=0$ for $\vec{x}\ne 0$, and thus the harmonic condition is preserved outside the source. It is also easy to check that (\ref{gaugefix}) fixes $\delta \tilde{h}_{00}=-\tilde{h}_{00}$ and $\delta \tilde{h}_{0i}=-\tilde{h}_{0i}$, and hence ensures that the transformed field $h^\prime_{\mu\nu}= h_{\mu\nu} + \delta h_{\mu\nu}$ has $h^\prime_{0\mu}=0$ everywhere. Furthermore,
\begin{align}\nonumber
\delta \tilde{h} &= -\delta \tilde{h}_{00} + \delta \tilde{h}_{ii} \\\nonumber&= \tilde{h}_{00} + \partial_i( 2 i\omega^{-1} \tilde{h}_{0i} - \omega^{-2} \partial_i \tilde{h}_{00})\\&=  - \omega^{-2} \partial_i^2 \tilde{h}_{00} - \tilde{h}_{ii},
\end{align}
where, in the last step, we have used the $\mu=0$ component of (\ref{FTHC}). Thus, for $\vec{x}\ne0$, where we may use (\ref{FTFE}), we have
\begin{align}
\delta \tilde{h} = \tilde{h}_{00} - \tilde{h}_{ii} = -\tilde{h},
\end{align}
so that $h^\prime=0$ outside of the source. In summary, the transformed field is
\begin{align}\label{TTtrans}
h^\prime_{ij} = \int^\infty_{-\infty} \frac{\ud \omega}{2\pi} e^{i\omega t}\left(\tilde{h}_{ij} + \frac{2i}{\omega} \partial_{(i}\tilde{h}_{j)0} - \frac{1}{\omega^2}\partial_i \partial_j \tilde{h}_{00} \right),
\end{align}
with all other components zero, and $h^\prime=0$, $\partial^\mu \bar{h}^\prime_{\mu\nu}=0$ everywhere outside the source.\footnote{It should now be clear why this method cannot be applied to the time-independent mode of the field: ill-defined contributions proportional to $\delta(\omega)/\omega$ or $\delta(\omega)/\omega^2$ would appear in the integral on the right-hand side of (\ref{TTtrans}). Even without a delta-function at $\omega=0$, this integral is not unambiguous until we explain how to deform the contour to avoid the poles there. We suggest the contour should dodge into the lower half of the complex plane, as this ensures that $h^\prime_{ij}(t_1)$ is dependent only on $h_{\mu\nu}(t_2)$ for $t_2 \le t_1 $, which is to say, the transformed field does not depend on future values of the untransformed field. Using this ``causal'' contour, we can substitute (\ref{FT}) into (\ref{TTtrans}) and perform the $\omega$ integral, arriving at $h^\prime_{ij} (t,\vec{x})= h_{ij} (t,\vec{x})+ \int^t_{-\infty} \ud t^\prime((t-t^\prime)\partial_i\partial_j h_{00}(t^\prime,\vec{x}) - 2 \partial_{(i}h_{j)0}(t^\prime,\vec{x}))$. In general, this formula is less useful than (\ref{TTtrans}), however it does reveal the asymptotic conditions that the dynamical field must obey for this gauge-transformation to be well-defined:  as $t\to -\infty$, the non-oscillatory modes of $\partial_i\partial_j h_{00}$ and $\partial_{(i} h_{j)0}$ must vanish faster than $t^{-2}$ and  $t^{-1}$ respectively.}

We are now in a position to apply this procedure to the gravitational field of the compact source (\ref{hpoint}). Before doing so, however, it is worth mentioning that the technique just described is not limited to compact sources. In generalising, the only adjustment needed is that (\ref{FTFE}) will only hold at $\vec{x}$ such that $T_{\mu\nu}(t,\vec{x})=0$ for all $t$. Figure \ref{TTdiagram} illustrates the difference between this technique and the standard method mentioned in section \ref{TTsource}.

\begin{figure}
\centering
\includegraphics[scale=.45]{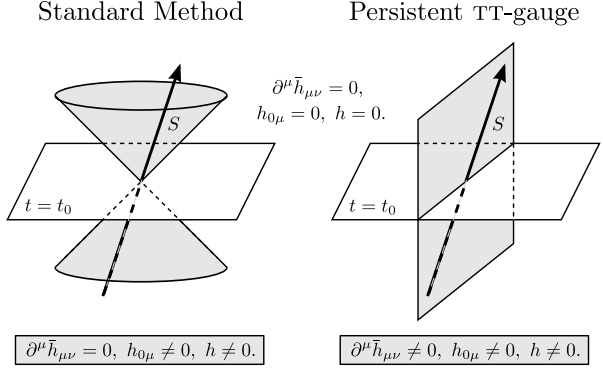}
\caption{Comparison of the standard method for achieving transverse-traceless gauge in the vicinity of a source \citep[chap. 4.4b]{Wald}, and the ``persistent'' method described here. In the two diagrams, $S$ represents an arbitrary source (a region with $T_{\mu\nu}\ne0$) moving relative to $u^\mu$. The hypersurface $t=t_0$ used to define the gauge in the standard method is also shown, but plays no role in our method.}
\label{TTdiagram}
\end{figure}

Continuing with the compact source, we write the dynamical part of (\ref{hpoint}) as
\begin{align}\bs
\bar{h}_{00} &= \partial_i \partial_j (\kappa I_{ij}(t-r)/4\pi r),\\
\bar{h}_{0i} &= \partial_j (\kappa \dot{I}_{ij}(t-r)/4\pi r),\\
\bar{h}_{ij} &=\kappa \ddot{I}_{ij}(t-r)/4\pi r,\es
\end{align}
and take the Fourier transform:
\begin{align}\bs
\tilde{\bar{h}}_{00} &= \tilde{I}_{ij} \partial_i \partial_j (\kappa e^{-i\omega r}/4\pi r),\\
\tilde{\bar{h}}_{0i} &= i \omega \tilde{I}_{ij}\partial_j (\kappa e^{-i\omega r}/4\pi r),\\
\tilde{\bar{h}}_{ij} &= - \omega^2 \tilde{I}_{ij} \kappa e^{-i\omega r}/4\pi r,\es
\end{align}
where $\tilde{I}_{ij}$ is the Fourier transform of the dynamical part of the quadrupole moment. Substituting this into (\ref{TTtrans}) yields
\begin{align}\nonumber
h^\prime_{ij} &= \int^\infty_{-\infty} \frac{\ud \omega}{2\pi} e^{i\omega t} \Bigg[\left( \tilde{I}_{kl}\delta_{ij} + \tilde{I} \delta_{ik}\delta_{jl} - 4\tilde{I}_{k(i}\delta_{j)l} \right)\partial_k \partial_l  \\\nonumber
&\quad \ \! \qquad \qquad \qquad- \frac{1}{\omega^2} \tilde{I}_{kl}\partial_k \partial_l \partial_i\partial_j  
\\
&\quad \ \! \qquad \qquad \qquad + \omega^2 \left( \tilde{I} \delta_{ij}-  2\tilde{I}_{ij}\right) \Bigg] \frac{\kappa e^{-i\omega r}}{8\pi r}.
\end{align}
Finally we recall that $h^\prime_{\mu\nu} = \bar{h}^\prime_{\mu\nu}$ (for $\vec{x}\ne 0$) and reinsert the time-independent mode
\begin{align}\nonumber
\langle \bar{h}_{00} \rangle &= \left(2M + \langle I_{ij} \rangle\partial_i\partial_j\right)\frac{\kappa}{4 \pi r},\\
\langle \bar{h}_{0i} \rangle &= J_{ij}\partial_j \frac{\kappa}{4 \pi r},\\\nonumber
\langle \bar{h}_{ij} \rangle &= 0,
\end{align}
to confirm equation (\ref{TThpoint}).

\newpage

\bibliography{LE}
\end{document}